\documentclass[a4paper,twocolumn,10pt,accepted=2024-11-21]{quantumarticle}
\pdfoutput=1

\usepackage{amsmath}
\usepackage{amssymb}
\usepackage{amsthm}
\usepackage{enumerate} 
\usepackage{subfigure} 
 \usepackage{hyperref}

\usepackage[margin=1in]{geometry}
\usepackage[utf8]{inputenc}
\usepackage[T1]{fontenc}

\usepackage{array}

\newcommand{\be}{\begin{equation}}
\newcommand{\ee}{\end{equation}}
\newcommand{\ba}{\begin{array}}
\newcommand{\ea}{\end{array}}
\newcommand{\bea}{\begin{eqnarray}}
\newcommand{\eea}{\end{eqnarray}}

\newcommand{\ent}{\mathrm{Ent}}

\def\>{\rangle} 
\def\<{\langle}
\def\wt{{\rm wt}}

\newenvironment{acknowledgement}%

\newtheorem{lem}{Lemma}
\newtheorem*{lem*}{Lemma}
\newtheorem{alg}{Algorithm}
\newtheorem{rem}{Remark}
\newtheorem*{rem*}{Remark}
\newtheorem{cor}{Corollary}
\newtheorem*{cor*}{Corollary}

\newtheorem*{fact*}{Fact}

\newtheorem{thm}{Theorem}
\newtheorem*{thm*}{Theorem}

\newcolumntype{C}[1]{%
 >{\vbox to 4ex\bgroup\vfill\centering}%
 p{#1}%
 <{\egroup}}

\begin{document}

\title{Constructing
quantum codes from any classical code and their embedding in ground space of local Hamiltonians}

\author{Ramis Movassagh}
\affiliation{IBM Research, MIT-IBM Watson AI lab, Cambridge MA, 02142, USA}\affiliation{Google Quantum AI, Los Angeles, CA, 90291, USA}
\email{movassagh@google.com}

\author{Yingkai Ouyang}  
\affiliation{School of Mathematical and Physical Sciences, University of Sheffield, Sheffield, UK 
}
\email{y.ouyang@sheffield.ac.uk}

\begin{abstract}
Implementing robust quantum error correction (QEC) is imperative for harnessing the promise of quantum technologies. 
We introduce a framework that takes {\it any} classical code and explicitly constructs the corresponding QEC code. 
Our framework can be seen to generalize the CSS codes, 
and goes beyond the stabilizer formalism~(Fig.~\ref{Fig:VennCSS}). 
A concrete advantage is that the desirable properties of a classical code are automatically incorporated in the design of the resulting quantum code. 
We reify the theory by various illustrations some of which outperform the best previous constructions.
We then introduce a local quantum spin-chain Hamiltonian whose ground space we analytically completely characterize.
We utilize our framework to demonstrate that the ground space contains explicit quantum codes with linear distance. This side-steps the Bravyi-Terhal no-go theorem.
\end{abstract} 
\maketitle

\section*{Overview}
In a way {\it information} is an ultimate basic entity in Nature. It is the information content that gives rise to the vast complexity of the physical world even though all known matter is built from a few subatomic particles. Information processing is a way by which we can view the physical processes, information preservation is how physical properties remain invariant in time, and error correction is how a physical state is restored when it  deviates from its original form. One can argue that better understanding of the information preservation and processing is the next frontier of science that bridges the micro and macro. 
\begin{figure*}
\centering
\includegraphics[scale=0.37
]{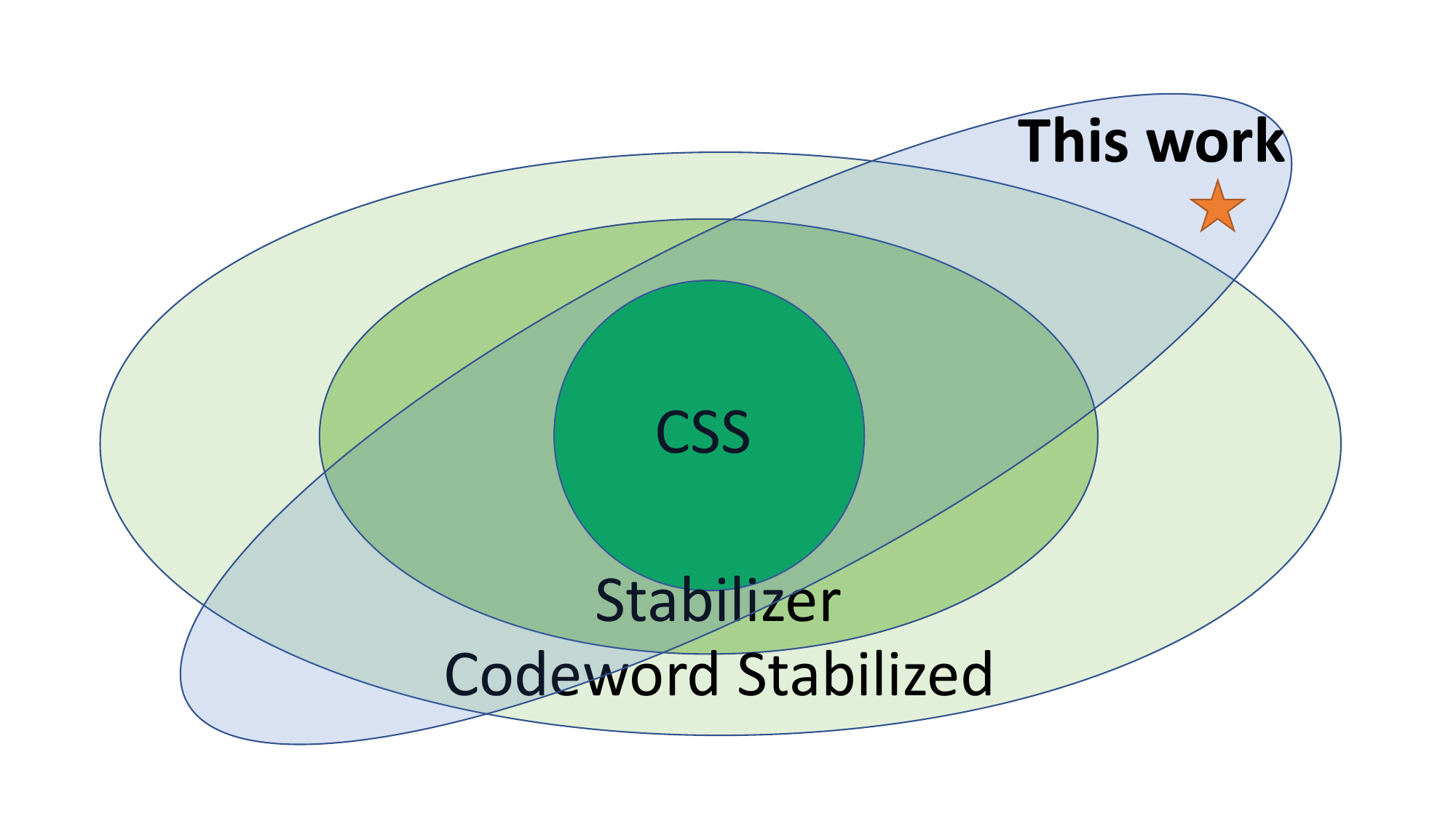} 
\caption{Comparison of the codes attainable within this work with  Calderbank-Shor-Steane (CSS)~\cite{calderbank1996good,steane1996multiple}, quantum stabilizer~\cite{gottesman96,CRSS1997quantum}, and codeword-stabilized (CWS)~\cite{CSSZ08} codes. Inclusion of stabilizers and CWS is not strict as our codewords are supported on disjoint sets. The star on the upper right indicates that we can identify codes outside of the previous frameworks.}
 \label{Fig:VennCSS}
\end{figure*}
Quantum theory is currently the best we have that accounts for the majority of physical phenomena. Therefore quantum information science beckons us towards a deeper understanding of the world. 

Computation is our way of synthesizing information processing for useful tasks. Therefore, quantum computation promises to be the most powerful model of computation available to us based on our current theories of nature. 

As stated above, preserving information and error correction are necessary parts of any reliable computation. On one hand, one can protect against errors by designing error correcting codes that allow for reliable recovery of the encoded information. On the other hand, nature herself provides innate resources for the protection and correction of information. 

Classical error correction, which applies to the standard (classical) model of computation is a mature and evolved field.
However, the successes there do not easily carry over to the quantum realm. The main difficulty is that possible quantum errors that can corrupt the computation or a state in the memory are richer than classical errors-- besides bit-flip errors that classical information suffers, there can be errors in the phases. 
Moreover, nature herself provides innate resources for the protection and correction of information.  
These beg the questions:\\
\noindent{\bf 1.} 
To what extent can one import discoveries in classical coding
 to the quantum realm? \\
{\bf 2.} Since the bulk of matter often resides in its ground state, are there simple physical quantum systems amenable to experimental realizations in the lab whose ground states sustain good quantum codes?\\

In this paper we address both of these questions in two parts, which we now overview: 

{\it Part 1} introduces a framework that takes any classical code and algorithmically constructs an explicit quantum code. 
The challenge in designing quantum codes is to not only correct the bit-flip errors, but also to correct phase-flip errors. 
We construct quantum codes with asymmetric distances given by $d_X$ and $d_Z$ for the bit-flip and phase-flip errors respectively.
Similar to CSS codes, the logical codewords are supported on classical codewords of length $n$. However, our formalism goes beyond the CSS and its generalization (the stabilizer formalism). We are able to identify codes outside the known paradigms. 
In the simplest case, we encode a logical qubit as follows
\begin{align}
    |0_L\> &= \sum_{{\bf c} \in C_0}  \alpha^{(0)}_{\bf c}|{\bf c}\> , \quad 
    |1_L\> = \sum_{{\bf c} \in C_1}  \alpha^{(1)}_{\bf c}|{\bf c}\> , \label{form1}
\end{align}
where $C_0$ and $C_1$ are disjoint subsets of $C$, and each $\boldsymbol{c}$ is a product state.
\begin{figure}
\centering
\includegraphics[scale=0.3
]{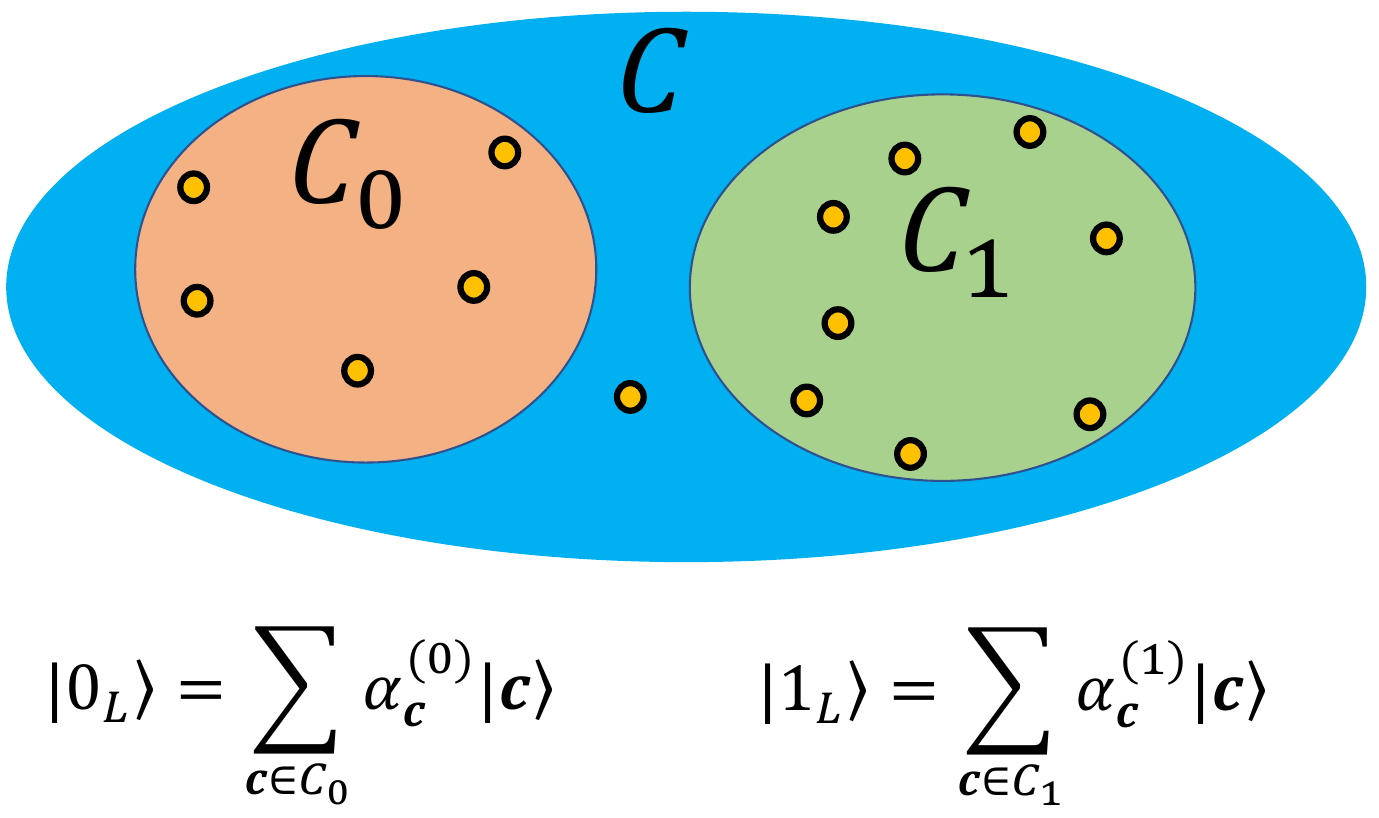} 
\caption{The encoding of one logical qubit from two disjoint subsets of a classical code $C$.}
 \label{Fig:LogicalQubit}
\end{figure}
{\it A priori}, the subsets $C_0$ and $C_1$ and the coefficients are not known. 
By mapping the Knill-Laflamme (KL) quantum error correction criteria~\cite{KnL97} to a linear algebra problem,
we construct an explicit algorithm that can determine these unknowns.

Using our framework we find constant excitation codes that outperform the best-known comparable codes \cite{ouyang2021avoiding,hu2021mitigating}. 

To design a quantum code that encodes more than one logical qubit, we provide a recursive algorithm. This algorithm finds the disjoint subsets of $C$ and enforces the KL criteria by solving for feasible solutions of a linear program. 
The algorithm outputs the logical states:
\begin{align}
    |j_L\> &= \sum_{{\bf c} \in C_j}  \alpha^{(j)}_{\bf c}|{\bf c}\> , \quad 
    j\in\{0,1,\dots,M-1\} , \label{formM}
\end{align}
where $M\le q^{c\,n}$ with $0<c<1$ resulting in quantum codes with constant rates, and $q$ is the dimension of each qudit.

Let $V_q(r)$ be the volume of the Hamming ball of radius $r$ defined by
\begin{align}
V_q(r) =\sum_{w=0}^r \binom n w (q-1)^{w}.
\end{align}
We prove (see Theorem~\ref{Thm:main1}):
\begin{thm*}
Take a classical code $C$ of length $n$ on a $q$-ary alphabet with 
the distance of $d_X$. 
If  $|C|\ge 2V_{q}(d_Z-1)$, then Alg.~\ref{alg:quBit} and Alg.~\ref{alg:quDit} in the paper explicitly derive $M$ logical states as in  Eq.~\eqref{formM} with bit- and phase-flip  distances of $d_X$ and $d_Z$ respectively. Further, 
if $C$ is chosen as a random code, with probability that approaches 1 as $n$ goes to infinity, we have
$2 \le M\le q^{c\:n}$ and $0<c<1$. The overall distance is $\min(d_X,d_Z)$.
\end{thm*}
Our algorithm
has a recursive structure. When the inequality in the theorem is satisfied, the algorithm always constructs two logical codewords. The third and subsequent logical codewords are found recursively by determining the feasibility of a sequence of linear programs. 

This recursive algorithm succeeds with high probability and almost surely over random codes. In the rare case that the set of linear constraints have rows that are all non-zero and are of the same sign, the recursion becomes infeasible and algorithm halts. 
We find that even if this happens we can always construct {\it Approximate} Quantum Error Correcting Codes (AQECC)~\cite{LNCY97} with linear distance and constant expected rate, provided that the underlying classical code has these parameters. The trade-off is between the expected number of logical states and the approximation error as given by Eq.~\eqref{eq:ExpectedLogical}.

In addition to giving asymptotic results, this formalism works equally well to construct finite size codes. To illustrate, we show examples of quantum codes that can be constructed from classical codes in the classical literature. Among the various, we find that the Steane code is the unique solution of our formalism when the input is Hamming's [7,4,3] code. The formalism introduced here is shown in relation to other known formalism in Fig.~\ref{Fig:VennCSS}.

Our framework allows one to have the following perspective on designing quantum codes: instead of choosing from existing quantum codes the best quantum codes that suits a given physical system, one can impose the constraints imposed by the physical system a priori, and then construct quantum codes that respect these constraints.
There is a big distinction between the two approaches of quantum code design. 
In the former approach, we might encounter a situation where the best available quantum code does not simultaneously have all the nice properties that we would like to have, whereas in this case, we flip the problem on its head.
Namely, once we can cast physical constraints in the language of classical coding theory, we have a systematic approach to design bespoke quantum codes that respects the required physical constraints. In the second part of our paper, we illustrate this with a concrete 2-local spin chain Hamiltonian, and design quantum codes for it. This illustrates a systematic approach  towards realising robust quantum codes in quantum matter.\\

{\it Part 2} focuses on the physics of information. Since topological models of quantum computation, it has been recognized that quantum codes may naturally appear in the ground space of physical systems~\cite{kitaev2003fault}. Local Hamiltonians are considered more physical than the general Hamiltonians. In particular the most physical are 2-local interactions, and many have investigated the encoding of quantum codes in their ground spaces.

The most celebrated example is Kitaev's toric code that resides in the ground space of a 4-local Hamiltonian, which is an effective Hamiltonian of a perturbed 2-local Hamiltonian~\cite{Kit06}.  
Kitaev's toric code is an example of topological order and paved the way for the topological model of quantum computation.  
The compass model~\cite{dorier2005quantum} is 2-local on a lattice, and has recently been proposed as a candidate to encode quantum codes in the eigenbasis of the Hamiltonian~\cite{li20192d}. However, the performance of these quantum codes are not well understood and have mainly been numerically investigated.
The advantage of our work over the aforementioned works is that we construct a 2-local Hamiltonian, whose quantum error correcting properties we analytically prove in the perfect quantum error correction setting.

Brandao {\em et al.}~\cite{brandao-PhysRevLett.123.110502} gave both constructive and non-constructive proof of the existence of AQECCs within the low-energy sector for a multitude of yo{quantum spin chains} including ferromagnetic Heisenberg model and spin-1 Motzkin spin chain~\cite{brandao-PhysRevLett.123.110502}. The challenges that remained were that the quantum error correction criteria was only approximately satisfied (hence AQECCs), errors had to be on consecutive set of spins for general spin chains, and the codes were in a low-energy sector of the local Hamiltonian (i.e., not the ground space). Moreover, the distance of the code grows logarithmically with the number of spins. 

Recently, it was also shown that by leveraging on a clock-work construction, one can design AQECCs with linear distance that reside within the ground-space of 10-local Hamiltonians \cite{bohdanowicz2019good}. However the challenge remained for constructing exact linear distance codes that reside within the ground space of a 2-local Hamiltonian. 

Our work overcomes these challenges. We construct explicit codes that exactly correct errors with linear distance that encode one logical qubit (we could have easily encoded a qudit as well). We write down a new and explicit 2-local quantum integer spin-$s$ chain parent Hamiltonian, $H_n$, on $n$ qudits. We analytically prove that its ground space can be spanned by product states. 
By mapping these product states to classical codewords, we reduce the problem of finding quantum codes in the ground space of our Hamiltonian to that of finding classical codes that must obey some constraints that are induced by the Hamiltonian. 
The classical coding problem becomes that of finding $q-$ary codes with forbidden sub-strings. 
By leveraging on existing constructions of binary codes, we construct candidate classical codes for our algorithm to run.
See the following table for  a comparison.
\begin{center}
\begin{tabular}{|c|c|c|}
\hline 
Properties & Brandao et al \cite{brandao-PhysRevLett.123.110502} & \!\!\!This work\!\!\!\tabularnewline
\hline 
\hline 
QECC & Apprx. $\epsilon=O(N^{-1/8})$ & Exact \tabularnewline
\hline 
Distance $d$ & \!\!\!$d=\Omega(\log(N))$\!\!\! & $d=\Theta(N)$\tabularnewline
\hline 
Rate & \!\!\!Vanishes\!\!\!  & Vanishes \tabularnewline
\hline 
\!\!\!Error restriction\!\!\! & consecutive spins & None\tabularnewline
\hline 
Code space & \!\!\!Low-energy eigenstates\!\!\! & \!\!\!Exact g.s.\!\!\!\tabularnewline
\hline 
\end{tabular}\medskip{}
\par\end{center}

We now give details of the main features of the Hamiltonian and its ground space.
Let us consider a spin chain of length
$n$ with open boundary conditions and the local Hilbert space dimension of $2s+1$, where $s\ge1$
is a positive integer. We take a representation in which $|j\rangle$
denotes the $s_{z}=j$ state of a spin-$s$ particle, such that $\hat{S}_{z}|j\rangle=j\;|j\rangle$ where $j\in\{0,\pm 1,\pm 2,\cdots,\pm s\}$.

The local Hamiltonian whose ground space contains the quantum code is $H_{n}=H_{n}^{J}+H_{n}^{s}$, where $H_{n}^{J}=J\sum_{k=1}^{n}\left(|0\rangle\langle0|\right)_{k}$.
The Hamiltonian $H_{n}^{s}$, is defined by
\begin{equation}
H_{n}^{s}=\sum_{k=1}^{n-1}\left\{ \sum_{m=-s}^{s}P_{k,k+1}^{m}+\sum_{m=1}^{s}Q_{k,k+1}^{m}\right\} \;,
\end{equation}
and the local terms are projectors acting on two neighboring spins $k,k+1$ are 
\begin{equation}
P^{m}=|0\leftrightarrow m\rangle\langle0\leftrightarrow m|,\\ Q^{m}=|00\leftrightarrow\pm m\rangle\langle00\leftrightarrow\pm m|,
\end{equation}
where $|0\leftrightarrow m\rangle \equiv \frac{1}{\sqrt{2}}\left[|0,m\rangle-|m,0\rangle\right]$,  $|00\leftrightarrow\pm m\rangle \equiv \frac{1}{\sqrt{2}}\left[|0,0\rangle-|m,-m\rangle\right]$,
and we denoted by $|j,\ell\rangle$ the spin state $|s_{k}^{z}=j,s_{k+1}^{z}=\ell\rangle$. We will be mostly interested in $s>1$.

Since $H_{n}^{s}$ is free of the sign problem (i.e.,~stoquastic),
the local projectors define an effective Markov chain, which  have
the following correspondence:

\begin{table}
\begin{center}
\begin{tabular}{|c|c|c|}
\hline 
Local Projector & Local Moves & Interpretation\tabularnewline
\hline 
\hline 
$P^{m}$ & $0m\longleftrightarrow m0$ & spin transport\tabularnewline
\hline 
$Q^{m}$ & $00\longleftrightarrow m,-m$ & spin interaction\tabularnewline
\hline 
\end{tabular}
\par\end{center}
\caption{Map between interactions and corresponding moves in the underlying Markov chain. Here {\it spin transport} means locally exchanging spin $m$ with $0$, and {\it spin interaction} means local creating or annihilating $m,-m$. }
\label{table:moves}
\end{table}
 
 We prove that the ground state degeneracy is exponentially large in the number of spins ($n\gg 1$):
\[
\dim(\ker(H_{n}^{s}))=\frac{\left(-2+r_{+}^{n+1}+r_{-}^{n+1}\right)}{2(s-1)} \approx \frac {r_{+}^{n+1}}{2(s-1)},
\]
where $r_\pm\equiv(1\pm\sqrt{1-1/s})$,
and $\ker(M)$ the kernel of the operator $M$.
Let us denote by $\ent_q : [0,1] \to [0,1]$  the $q$-ary entropy function defined by $\ent_q(x)= -x\log_{q}x-(1-x)\log_{q}(1-x)+x\log_{q}(q-1)$.
Our main theorem is (this is Theorem \ref{thm:linear-distance} in the paper):
\begin{thm*}
Let $0<\tau\le 1/4$ be a real and positive constant.
There exist quantum codes in $\ker(H_n)$ that encode one logical qubit and have the distance of $2 \tau n$ whenever 
\begin{align}
\ent_2(2\tau) + \ent_{2s+1}(2\tau) \log_2 (2s+1) + o(1) \le 1.\nonumber 
\end{align}
Second, there are explicit quantum codes which encode one logical qubit with a distance of $2 \tau n$ whenever
\begin{align}
1/2-\tau/0.11  \ge \log_2(2s+1) \ent_{2s+1}(2\tau).\nonumber
\end{align}
\end{thm*}
\begin{rem*}
We call the constructions given by optimizing \eqref{eq:gv-construct} and \eqref{eq:j-construct} as the Gilbert-Varshamov (GV)~\cite[Chpt.~1]{sloane} and Justesen construct~\cite[Chpt.~10,~Thm.~11]{sloane} respectively.
The GV construct arises from choosing a random $C$, while the Justesen construct uses the classical Justesen code to define $C$.
\end{rem*}

Our work side-steps the Bravyi-Terhal no-go theorem~\cite{BrT09} which asserts that in one-dimension one cannot have stabilizer codes with system-size dependent distance. 
We also side-step the more general no-go result given by Gschwendtner {\rm et al.} \cite{gschwendtner2019quantum} which states that the degenerate ground space of gapped Hamiltonians can only be QECs with a constant distance.
The reason we can side-step these no-go results is because our work  considers QEC as a strict subspace of the ground space of Hamiltonians.
In particular, this work could pave the way for constructing the quantum codes with linear distance within the ground space of translation invariant local spin chains. We note that had we used {\it all} of the ground space to construct the codes, 
this Hamiltonian could have made the case for the first example of topological order in one-dimension, which has been conjectured to be impossible.  The practical advantage of our work is that such explicit Hamiltonians are easily constructed in the laboratory in the near term, especially in atomic or ion trap architectures. Lastly, the Hamiltonian is a generalization of the highly entangled colored Motzkin spin chain~\cite{movassagh2016supercritical}, which may be of independent interest.

\subsection{Practical considerations} 

Here, we explain why the quantum codes we can construct, particularly the non-stabilizer codes, are interesting from a practical point of view.

A particular family of non-stabilizer codes that our paper can construct, and also show potential in a practical setting, are permutation-invariant quantum codes.
Permutation-invariant codes can be efficiently prepared using either the usual Boykin gate set \cite{boykin} (Clifford + T gates) set \cite{bartschi2019deterministic}, or quantum control type of operations such as Rabi oscillations \cite{init-picode-2019-PRA}, or geometric phase gates \cite{johnsson2020geometric}. 
In particular, geometric phase gates can be performed in a very practical way, because they only need (1) a linear interaction between the collective angular momentum operator of the qubits and a bosonic mode, and the bosonic mode to be initialized in a coherent state (a laser mode). 
The paper \cite{johnsson2020geometric} showed that geometric phase gates can be performed in variety of physical systems. Good candidate systems for geometric phase gates include ion traps and cold neutral atoms.

Furthermore, there is potential to do quantum error correction on permutation-invariant codes in a practical way, as described in Ref.~\cite{ouyang2022quantum}. Ref.~\cite{ouyang2022quantum} proposes to perform quantum error correction by (1) measuring the total angular momentum of sequences of consecutive qubits, (2) performing geometric phase gates, and (3) performing logical gate teleportation operations. Ref.~\cite{ouyang2022quantum} explains how each one of these operations can be performed in a practical way, by using the interaction between collective spin operators and one or more bosonic modes, combined with homodyne and heterodyne measurements on these bosonic modes. 
Hence permutation-invariant codes have strong potential to be practical codes.
Hence our paper's formalism does encompass codes with practical potential.


Our paper's code construction formalism also encompasses a particular class of quantum codes that reside in a decoherence free subspace of errors that act coherently \cite{ZaR97,alber2001stabilizing,alber2003detected,bacon2000universal}, and are hence `error-avoiding' quantum codes. 
There has been recent work on constant-excitation codes, which are quantum error correction codes that avoid coherent errors in one axis (in Refs.~\cite{ouyang2021avoiding,hu2021mitigating}). 
Using the code construction formalism in our paper and also by considering non-stabilizer codes, we construct new constant-excitation codes with improved distance-rate parameters for fixed code length.

We leave questions regarding the fault-tolerance of quantum codes that can be constructed using our framework as an open problem, and we are optimistic about its ultimate resolution.

\subsubsection{Why introduce a Hamiltonian?}
{\it Theoretical considerations}-- This model is of theoretical interest because it allows for the encoding of linear distance codes in its ground space. Encoding quantum information in the ground space of physical Hamiltonians has a long history dating back to Kitaev's toric code~\cite{kitaev2003fault}. A motivation for our work was the nice work of Brandao {\em et al.}~\cite{brandao-PhysRevLett.123.110502} where existence of codes in low energy eigenstates of local translation invariant spin chains were found. In this work we came up with explicit and exact constructions and overcome some of the limitations of that work. These were detailed in the introduction and will be elaborated on below.  


Our work may also be useful in bringing new insights into conformal field theories from quantum codes  \cite{dymarsky2020quantum}, by allowing any classical code to be used in the quantum coding framework.

{\it Applied and engineering considerations}--
Finding new quantum codes can help hasten the dream of fault-tolerant quantum computation. Our method is distinct from other widely used methods to construct quantum codes from classical codes, such as for CSS codes, stabilizer codes, and codeword-stabilized codes. A key feature of our code construction is that we can take as input a general classical code, and demand that the quantum code we construct must be supported on the computational basis vectors that are labelled by these classical codewords. The non-trivial solution of the nullspace then gives us the amplitudes over with the logical zero and logical one of our quantum code will assume.

Our proposal to engineer a two-local Hamiltonian to stabilize quantum information is in line with the ideas utilizing quantum control techniques to suppress the noise before employing quantum error correction. We differ from traditional approaches in quantum control procedures where one typically applies dynamical decoupling pulses to create an essentially a trivial identity Hamiltonian that acts on the system when no quantum gates are performed. In this situation, however, local errors are not energetically penalized. In contrast, one might envision that quantum control methods can engineer a Hamiltonian that energetically penalizes the dominant noise rates that occur in a quantum system before introducing quantum error correction~\cite{PhysRevLett.113.260504,ouyang2019mems}.  

In many practical physical systems, noise is biased and can be dominated by bit-flip or phase-flip type errors. 
We consider a biased noise model that is dominated by bit-flip type errors. 
On a spin-system, we expect our Hamiltonian to energetically penalize bit-flip errors. This would allow our engineered Hamiltonian to greatly suppress the noise rates of the dominant (bit-flip) type of errors. The remaining errors can then be cleaned up using our quantum code with a linear distance. The advantage of engineering our Hamiltonian, as compared to, for instance, the surface code Hamiltonian, is that the Hamiltonian terms that we require are two-local, whereas the surface code requires many-body interactions, which is challenging to realize in practice.

The Hamiltonian we will introduce has the form $H_n=\sum^{n-1}_{k=1} H_{k,k+1}$; see Eq.\eqref{eq:Hn_s-2}. Its ground space has a strict subspace that is the quantum code. We could make the model sums of commuting local terms by considering a new Hamiltonian that skips all (say) even interactions and write $H'_n=\sum_{k\;odd}H_{k,k+1}$. The ground space of $H'_n$ contains the ground space of $H_n$ and therefore also includes the quantum code. One could continue this way and eventually get a Hamiltonian that is trivial ({\em i.e}., no interaction) $H''=I$ for which any quantum code is in the ``ground space''. The price one pays following this crooked path is that Nature will help less and less in suppressing the rate at which errors appear.\\


\subsection{Discussions and open problems}
This paper provides a rigorous framework for the systematic construction of a quantum codes from any classical code. We illustrate the theory through a series of examples and proved that new quantum codes with linear distance and constant rate can be constructed using this work.
Our formalism encapsulates the CSS formalism, and has an intersection with stabilizer and codeword-stabilized  (CWS) formalisms (see subsection \ref{sec:nonlinear-cyclic}). However, there are codes inside the stabilizer and CWS formalism that our formalism does not capture. For example, the five-qubit code is not covered by our formalism because the corresponding classical code has a distance of one. Alg.~\ref{alg:quDit} can construct logical states beyond a logical qubit. The relation of our work to the previous is faithfully depicted in the Venn diagram (Fig.~\ref{Fig:VennCSS}) shows.

Our formalism can take as the input any classical linear, self-orthogonal code, and then derive the corresponding CSS code. This is because CSS codes, which rely on these types of classical codes, can also be completely cast within our framework.
It will be interesting to classify all quantum codes that can be obtained within our framework if one were to start with classical linear, self-orthogonal codes.

Another open problem would be to find an alternative way of getting at our logical qu$d$it construction by directly using the large freedom in choosing codes in our construction (namely the high-dimensional kernel of the matrix $A$ in Eq.~\eqref{eq:Amatrix} and Fig.~\ref{fig:Ax}).
It would also be interesting to see an extension of our formalism to encompass all stabilizer codes, and indeed, any arbitrary quantum code. With the exception of permutation invariant codes (subsection \ref{subsec:permInva}), most of the analysis herein takes the classical codewords and uses them to define a product basis over which the logical quantum states are defined. The extension of our results to include non-product basis for the logical codewords calls for further investigation.


Another interesting problem is to find a Hamiltonian whose ground space is quantum code with a macroscopic distance, and the local and global ground states satisfy a consistency criterion as defined in~\cite{bravyi2011short}. This would then serve as the first example of topological quantum order in one-dimension. We would have found it easy to prove a gap above the degenerate ground space of the local Hamiltonian herein; however, since the code occupies a subspace of the ground space one would need to prove a {\it local gap} lower bound. This means that in order to move from a subspace of the ground space to another subspace an operator with a large support needs to be applied.

\section{Part~1:~Explicit quantum codes from classical codes}

\subsection{Constructing a logical qubit with linear distance}
In this section, we want to design $q$-ary quantum codes with bit-flip distance of $d_X$ and a phase-flip distance of $d_Z$ using a $q$-ary classical code $C\subset \{0,1\dots, q-1\}^n$ as an input.  The minimum distance of the quantum code is then $d = \min(d_X, d_Z)$. 

To correct errors on $q$-ary quantum codes, we consider errors in the generalized Pauli basis. Let us denote by $\omega$ the primitive root of unity $\omega\equiv \exp(2\pi i/q)$. The $Z$ and $X$ type Pauli matrices are respectively given by 
\begin{equation}\label{eq:PauliXZ}
Z = \sum_{j=0}^{q-1}\omega^j\:
|j\>\<j|\;,\quad X = \sum_{j \in \mathbb Z_q} |j\>\<j + 1|\:.
\end{equation}
And any Pauli operator is equivalent to 
$X^a Z^b$ up to a phase for some $a,b = 0,\dots, q-1$.
We consider the set Pauli operators on $n$ qudits $\mathcal P_n = \{X^a Z^b :a,b =0,\dots, q-1 \}^{\otimes n}$ that span the space of linear operators on $n$ qudits.
Given any Pauli in $\mathcal P_n$ that has the form $P = X^{a_1} Z^{b_1} \otimes \dots \otimes X^{a_n} Z^{b_n}$, we denote 
$\wt_X(P) = \wt ({\bf a} )$
 and 
 $\wt_Z(P) = \wt ({\bf b} )$,
 where 
 ${\bf a} = (a_1,\dots, a_n)$, ${\bf b} = (b_1,\dots, b_n)$, and 
 $\wt(\cdot)$ denotes the Hamming weight of the vector. 
 It is easy to see that for any non-negative integer $r \le n$, we have $\sum_{w=0}^{r} |\mathcal Z_w| = V_q(r)$, where $V_q(r) =\sum_{w=0}^r \binom n w (q-1)^{w}$
is the volume of the $q$-ary Hamming ball of radius $r$ as before.

For the KL criteria to hold for a quantum code with logical codewords $|0_L\rangle$ and $|1_L\rangle$ on $n$ qubits with a bit-flip distance of $d_X$ and a phase-flip distance of $d_Z$, it suffices to require that for all generalized Pauli matrices $P$ such that $\wt_X(P)\le d_X-1$ and $\wt_Z(P)\le d_Z-1$, these hold:
\begin{eqnarray}
\langle0_{L}|P|0_{L}\rangle\: & = & \:\langle1_{L}|P|1_{L}\rangle\label{eq:nonDef}\\
\langle0_{L}|P|1_{L}\rangle\: & = & 0\:.\label{eq:orth}
\end{eqnarray}
Eqs.~\eqref{eq:nonDef} and~\eqref{eq:orth} are the \textit{non-deformation}
and \textit{orthogonality} conditions respectively. If we demand
that 
\begin{enumerate}
\item $\text{dist}({C})\ge d_X\;$: the minimum distance of ${C}$
is at least $d_X$
\item $\text{Supp}(0_{L})\cap\text{Supp}(1_{L})=\emptyset\;$: (the logical
codewords $|0_{L}\rangle$ and $|1_{L}\rangle$ are supported on distinct
codewords in ${C}$),
\end{enumerate}
then the orthogonality condition (Eq.~\eqref{eq:orth}) trivially holds.
To verify the non-deformation condition (Eq.~\eqref{eq:nonDef}), we
note that $\langle0_{L}|P|0_{L}\rangle=\langle1_{L}|P|1_{L}\rangle=0$
whenever $P$ is not diagonal. Hence the only non-trivial cases to
be verified are the diagonal generalized Pauli operators, where the set of diagonal Pauli operators of weight
$w$ is 
\begin{equation}
\mathcal{Z}_{w}=\left\{ Z^{z_{1}}\otimes\cdots\otimes Z^{z_{n}}\;:\;\text{wt}({\bf z})=w\right\}.\label{eq:Zw}
\end{equation}

In general, for any diagonal Pauli operator $P$, the expectations
$\langle0_{L}|P|0_{L}\rangle$ and $\langle1_{L}|P|1_{L}\rangle$
are complex numbers. This is in contrast to the case where $P$ are
Kraus operators of the amplitude damping channel, in which all such
expectations are real, or when $P$ are diagonal operators and $q=2$. For Eq.~\eqref{eq:nonDef} to hold, the following has to hold for all Paulis $P$ with a weight at most $d-1$:
\begin{align*}
\text{Re}\left(\langle0_{L}|P|0_{L}\rangle\right)-\text{Re}\left(\langle1_{L}|P|1_{L}\rangle\right) & =0\\
\text{Im}\left(\langle0_{L}|P|0_{L}\rangle\right)-\text{Im}\left(\langle1_{L}|P|1_{L}\rangle\right) & =0\;,
\end{align*}
The quantum code we define
depends on a ``balanced" real non-zero column vector ${\bf x}=(x_{1},x_{2},\dots,x_{m})^T$ 
in the sense that
\[
\sum_{i=1}^{m}x_{i}=0.
\]
Let $x_{k}^{+}=\max\{x_{k},0\}$, 
$x_{k}^{-}=-\min\{x_{k},0\}$, 
and
$x=x_{1}^{+}+\dots+x_{m}^{+}$. 
We can decompose the vector ${\bf x}$
into its positive and negative components ${\bf x}={\bf x}^{+}-{\bf x}^{-}$
where ${\bf x}^{+}=(x_{1}^{+},\dots,x_{m}^{+})^T$ and ${\bf x^{-}}=(x_{1}^{-},\dots,x_{m}^{-})^T$
respectively. We also have $x=\left\Vert {\bf x}^{+}\right\Vert _{1}=\left\Vert {\bf x}^{-}\right\Vert _{1}=\left\Vert {\bf x}\right\Vert _{1}/2$.
For example, using this notation ${\bf x}=(1,2,-1,-1,-1)^T$ gives 
$\smash{{\bf x^{+}}=(1,2,0,0,0)^T}$
and $\smash{{\bf x}^{-}=(0,0,1,1,1)^T}$, and $x=3$. 

In our construction of quantum codes using the classical code ${C}$, we only consider logical codewords that are linear combinations over labels in ${C}$ with only real coefficients. 
Hence, we define the two logical codewords of our quantum code as 
\begin{eqnarray}
|0_{L}\rangle & = & \frac{1}{\sqrt{x}}\left(\sqrt{x_{1}^{+}}|{\bf c}_{1}\rangle+\cdots+\sqrt{x_{m}^{+}}|{\bf c}_{m}\rangle\right),\label{eq:Logical0}\\
|1_{L}\rangle & = & \frac{1}{\sqrt{x}}\left(\sqrt{x_{1}^{-}}|{\bf c}_{1}\rangle+\cdots+\sqrt{x_{m}^{-}}|{\bf c}_{m}\rangle\right).\label{eq:Logical1}
\end{eqnarray}
Since that $x^+_j x^-_j = 0$ for all $j\in[m]$,  the logical states 
$|0_L\>$ and $|1_L\>$  have disjoint supports.

We next clarify the connection between ${\bf x}$ and the non-deformation
conditions by constructing a real matrix $A$ that enforces these conditions. Roughly speaking, this matrix has
rows labeled by diagonal Pauli errors of weight at most $d_Z-1$ and
columns labeled by the states $|{\bf c}_{1}\rangle,\dots,|{\bf c}_{m}\rangle$.
While the ordering of the rows of $A$ is unimportant, we will collect
the rows in groups corresponding to the weights of $P$. The matrix
$A$ is defined by
\begin{align}
A=&\sum_{P\in\mathcal{Z}_{0}}\sum_{k=1}^{m}|P\rangle\langle k|
\notag\\&
+
\sum_{w=1}^{d-1}
\sum_{P\in\mathcal{Z}_w}
\sum_{k=1}^{m}\left\{ \text{Re}\left(\langle{\bf c}_{k}|P|{\bf c}_{k}\rangle\right)|P,0\rangle\langle k|
\right.\notag\\&\left.\quad\quad\quad\quad\quad\quad
+\text{Im}\left(\langle{\bf c}_{k}|P|{\bf c}_{k}\rangle\right)|P,1\rangle\langle k|\right\} .\label{eq:A_first}
\end{align}
In matrix representation $A$ is a wide rectangular matrix and writes
\begin{align}
A=
\begin{bmatrix}
    1 & \cdots & 1\\
a_{2,1} & \cdots & a_{2,m}\\
\vdots & \cdots & \vdots\\
a_{2V_{q}(2t)-1,1} & \cdots & a_{2V_{q}(2t)-1,m}\\
\end{bmatrix} \equiv
\begin{bmatrix}
   {\bf a}_{1}^{T}\\
{\bf a}_{2}^{T}\\
\vdots\\
{\bf a}_{2V_{q}(2t)-1}^{T} 
\end{bmatrix} \label{eq:Amatrix}
\end{align}
where ${\bf a}_{r}=(a_{r,1},\dots,a_{r,m})^T$ are column vectors. The reason for introducing the matrix $A$ is
that the non-deformation condition for correcting $t$ errors using
the vector ${\bf x}$ is enforced by the constraint (see Fig.\ref{fig:Ax})
\begin{figure*}
    \centering
    \includegraphics[scale=0.40]{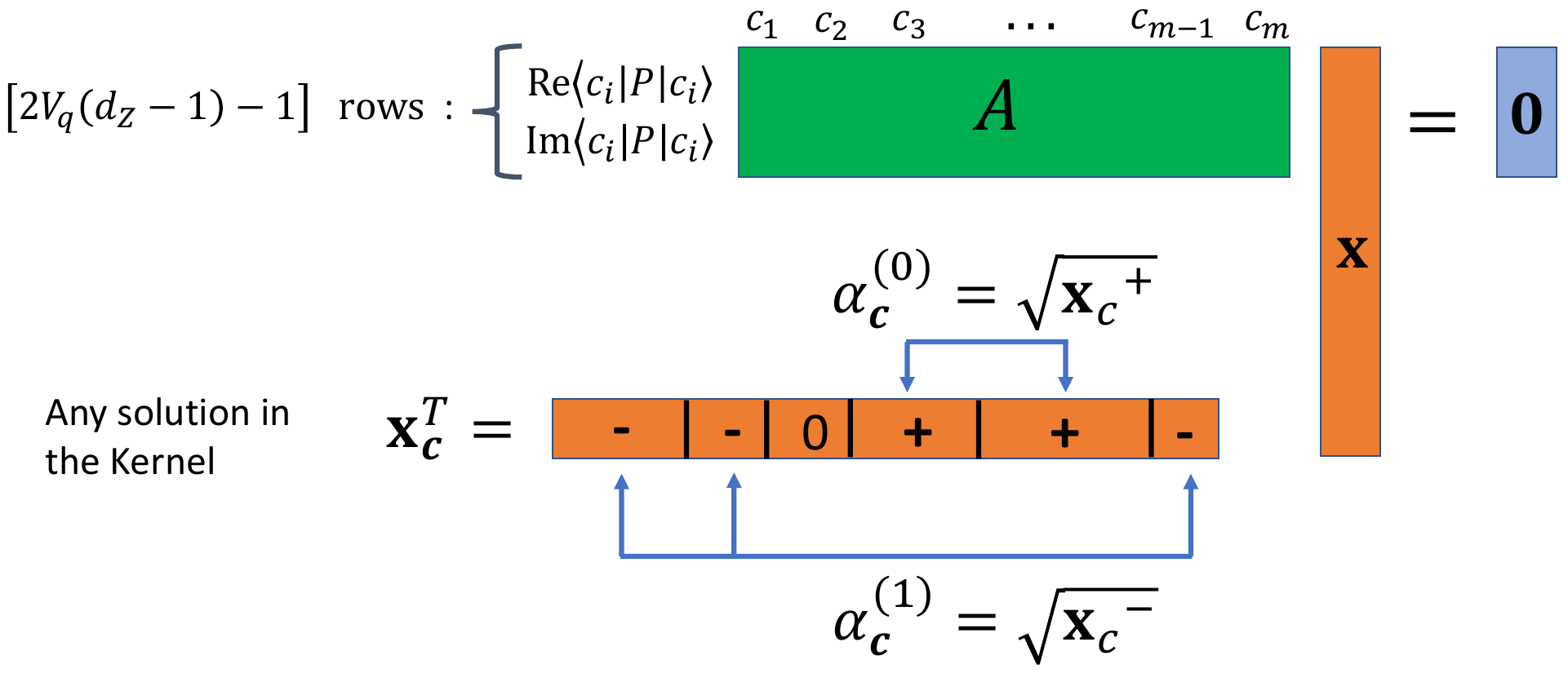}
    \caption{Illustration of $A{\bf x}=0$ in Lemma 1, where the matrix $A$ has $[2V_q(d_z-1)-1 ]$ rows and $m$ columns (top of this figure). The vector ${\bf x_c}$ has $m$ real components and denotes any solution in the kernel. Since the first row of $A$  is all ones ($P$ corresponds to the identity matrix), the vector ${\bf x_c}$ is balanced-- the sum of all entries is equal to zero.}
    \label{fig:Ax}
\end{figure*}
\[
A{\bf x}=0\;.
\]
\begin{lem}\label{Lem:AxZer0}
Let ${\bf x}$ be a non-zero real vector such that $A{\bf x}=0$.
Let $|0_{L}\rangle$ and $|1_{L}\rangle$ be logical codewords that
depend on ${\bf x}$ as in Eqs.~\eqref{eq:Logical0} and \eqref{eq:Logical1}.
Then $\langle0_{L}|P|0_{L}\rangle=\langle1_{L}|P|1_{L}\rangle$ for
any diagonal Pauli of weight at most $d_Z-1$. 
\end{lem}
\begin{proof}
Recall that $x=x_{1}^{+}+\dots+x_{m}^{+}$. Since each $|{\bf c}_{k}\rangle$ is a product state, we have $\langle{\bf c}_{j}|P|{\bf c}_{k}\rangle=0$
for all distinct $j$ and $k$, and for all diagonal Pauli of weight at most $d_Z-1$. 
We can use the definitions of logical codewords (Eqs.~\eqref{eq:Logical0}
and \eqref{eq:Logical1}) to write

\begin{align*}
&x\left(\langle0_{L}|P|0_{L}\rangle-\langle1_{L}|P|1_{L}\rangle\right) 
\notag\\
= & \sum_{k=1}^{m}x_{k}^{+}\langle{\bf c}_{k}|P|{\bf c}_{k}\rangle-\sum_{k=1}^{m}x_{k}^{-}\langle{\bf c}_{k}|P|{\bf c}_{k}\rangle\\
 = & \sum_{k=1}^{m}x_{k}\langle{\bf c}_{k}|P|{\bf c}_{k}\rangle
\end{align*}
where on the second line we used $x_{k}=x_{k}^{+}-x_{k}^{-}$. Using
Eq.~\eqref{eq:A_first} we see that $\langle{\bf c}_{k}|P|{\bf c}_{k}\rangle
=
\langle P,0|A|k\rangle
+
i \langle P,1|A|k\rangle$.
Therefore
\begin{align*}
&x\left(\langle0_{L}|P|0_{L}\rangle-\langle1_{L}|P|1_{L}\rangle\right) 
\notag\\
= & \sum_{k=1}^{m}x_{k}\left(\langle P,0|A|k\rangle+
i \langle P,1|A|k\rangle\right)\\
  = & \sum_{k=1}^{m}\left(\langle P,0|Ax_{k}|k\rangle+
 i \langle P,1|Ax_{k}|k\rangle\right)\\
  = & \langle P,0|A{\bf x}
 + i\langle P,1|A{\bf x}\;,
\end{align*}
because ${\bf x}=\sum_{k=1}^{m}x_{k}|k\rangle$. Since the first row
of $A$ is all ones, $A{\bf x}=0$ implies $\smash{\langle0_{L}|0_{L}\rangle=\langle1_{L}|1_{L}\rangle.}$
Moreover, in the above, $P$ is an arbitrary diagonal Pauli of weight at most $d_Z-1$,
and the requirement $A{\bf x}=0$ implies $\langle0_{L}|P|0_{L}\rangle=\langle1_{L}|P|1_{L}\rangle$
for any diagonal Pauli of weight at most $d_Z-1$. 
\end{proof}
The condition $A{\bf x}=0$ is satisfied for any ${\bf x}\in\ker(A)$.
It is then important to understand the structure of the kernel. Since
$A$ is real, any vector in its kernel must be real, which
we then use to design an explicit quantum code that obeys the generalized
orthogonality conditions above. Whenever $|{C}| \ge 2V_{q}(d_Z-1) $, that
is the number of codewords in ${C}$ is strictly greater than the number of rows in $A$, 
by the rank-nullity theorem, $A$ must have a non-trivial kernel. We thus have the following existence theorem for quantum codes that reside within the ground space of $H_n$.
\begin{lem}\label{lem:main-result}
(Existence) Let ${C}$ have a minimum distance at least $d_X$. If $|{C}|\ge 2V_{q}(d_Z-1)$,
then Alg.~\ref{alg:quBit} constructs a quantum code with one logical qubit and with a bit-flip and phase-flip distance of $d_X$ and $d_Z$ respectively. Moreover, this quantum code can be constructed from any nonzero ${\bf x}\in\ker(A)$.
\end{lem}
 This lemma shows that we can derive the logical codewords of our quantum code from a classical code, and only requires ${C}$ to have a minimum distance of $d_X$ and for $|{C}|$ to be at least twice the size of the $q$-ary Hamming ball of radius $d_Z-1$. 

We proceed to derive a more explicit expression for the entries of
the matrix $A$ in Eqs.~\eqref{eq:A_first} and \eqref{eq:Amatrix}. 
We will show that these entries are elements of a finite-sized set given by
\begin{equation}
\{\cos(2\pi k/q)\;,\;  \sin(2\pi k/q) \::\:k = 0,\dots,q-1\}\;.\label{eq:gamma}
\end{equation} 
For any classical string ${\bf c}=(c_{1},\dots,c_{n})^T\in{C}$ with $c_i = 0,1,\dots,q-1$, the quantum state is $|{\bf c}\rangle=|c_1,c_2,\dots,c_n\rangle$, and 
when $P = Z^{z_1}\otimes \dots \otimes Z^{z_n}$, we have
\begin{align}
    &\<{\bf c}|P|{\bf c}\> 
    = \prod_{j=1}^n \<c_j|Z^{z_j}|c_j\>\notag\\
    =& \prod_{j=1}^n \<c_j|\sum_{k\in\Sigma_c} (\omega^k)^{z_j} |k\>\<k|c_j\>\\
    =&\prod_{j=1}^n \omega^{z_j c_j} =\omega^{{\bf z}^{T}{\bf c}}
    \\
    =&\cos\left(\frac{2\pi\,{\bf z}^{T}{\bf c}}{q}\right)+i\sin\left(\frac{2\pi\,{\bf z}^{T}{\bf c}}{q}\right).\label{eq:cosine-sine}
\end{align} 
The first row of $A$ is all ones, and all other entries are given
by
\begin{align}
A'&=\sum_{w=1}^{d_Z-1}
\sum_{\wt({\bf z}) = w}
\sum_{{\bf c}\in{C}}
\left[
\cos\left(\frac{2\pi\,{\bf z}^{T}{\bf c}}{q}\right)|{\bf z},0\rangle\langle{\bf c}|\right.\notag\\
&\qquad\qquad\left.+
\sin\left(\frac{2\pi\,{\bf z}^{T}{\bf c}}{q}\right)|{\bf z},1\rangle\langle{\bf c}|
\right].
 \label{eq:Aprime}
\end{align} 
Since ${\bf z}^{T}{\bf c}$ is always an integer, it follows that
the entries of $A$ must take values from the set in \eqref{eq:gamma}. 
\begin{rem}
In Eq.~\eqref{eq:Aprime} finding a solution in the kernel amounts to finding linear combination of roots of unity that vanish. See for example section III and references in~\cite{movassagh2017green}.   There is a vast literature (see chapter 7 in~\cite{sloane}) on this topic and properties of the underlying code that controls the values of the integers ${\bf z}^{T}{\bf c}$ can in principle be utilized to give analytic solutions. 
\end{rem}

\begin{rem}
When $q=2$, all rows in $A'$ that are labeled by $|{\bf z},1\>$ are equal to zero, because the argument of the sine is always an integer multiple of $\pi$.
\end{rem}

This section is summarized in the following algorithm:

\begin{alg}\label{alg:quBit} {\bf Input:} A classical code $C\subset \{0,1,\dots, q-1\}^n$ with $m\equiv |C|$.
\begin{itemize}
    \item Form the matrix $A$ defined by Eq.\eqref{eq:A_first}.
    \item Solve $A{\bf x} = 0$ to find  ${\bf x}\ne 0$. Define ${\bf x^+},{\bf x^-}\ge 0$ such that ${\bf x=x^+-x^-}$ as in Lemma \ref{Lem:AxZer0}
    \item Let $C_0$ be the set of codewords ${\bf c}_k$ with $k\in supp({\bf x^+})$. And let $C_1$ be the set of codewords ${\bf c}_k$ with $k\in supp({\bf x^-})$. 
    \item For all ${\bf c} \in C_0$, assign $\alpha^{(0)}_{\bf c} = \sqrt{x^+_c}$, and for all ${\bf c} \in C_1$, assign $\alpha^{(1)}_{\bf c} = \sqrt{x_c^-}$. 
\end{itemize}
{\bf Output:} A logical quantum bit with codewords $|0_L\rangle$ and $|1_L\rangle$ as defined in Eqs.~\eqref{eq:Logical0} and \eqref{eq:Logical1}. The code distances are  $d_X=\text{dist}(C)$ and $d_Z$ that satisfies $ m>2V_q(d_Z-1)-1$.\\
\end{alg}

\subsubsection{Limitations of the formalism}
\label{sec:6-3-3-code}

Can we use the classical repetition code with codewords 0\dots 0 and 1\dots 1 to make a QEC code with distance 2? 
There are only two linearly independent columns in the matrix $A$. Hence, we can only have a non-trivial kernel if we have less than two rows. This is not sufficient for even detecting a single  error, which would require a distance to be equal to two. 
This is consistent with the fact that the quantum repetition code, with codespace spanned by $|0\dots 0\>$ and $|1\dots 1\>$, cannot deal with phase errors.

Next, we consider a classical code $C$, and like to design a corresponding QEC code with maximum distance.
As an example, we consider a classical code that comprises of the codewords
000000,
001110,
010101,
011011,
100011,
101101,
110110,
and
111000. 
We like to know if our framework allows us to obtain a QEC code with distance 3.
However, when we construct the $A$ matrix using our formalism, we find that its nullspace has dimension 0. 
Hence, there is no QEC code with support on computational states labelled by $C$ with distance 3.

However if the designed distance is reduced to 2, then we can use our formalism to derive the following error-detecting quantum code
\begin{align}
|0_L\> &=
\frac{1}{2}\left(|0,0,0,0,0,0\>+
|0,1,1,0,1,1\>
\right.\notag\\&\quad\left.+
|1,0,1,1,0,1\>+
|1,1,0,1,1,0\>\right)
\\
|1_L\> &= 
\frac{1}{2}\left(
|0,0,1,1,1,0\>+
|0,1,0,1,0,1\>\right.\notag\\&\quad\left.+
|1,0,0,0,1,1\>+
|1,1,1,0,0,0\>\right)
\end{align}

\subsection{Constant-excitation codes with better parameters}
When spurious classical fields interact uniformly with a system of spins, an effective unitary evolution $U_\theta= \exp(i\theta(Z_1 + \dots Z_n))$ acts on the spins. 
Here $Z_i$ applies a phase flip on the $i$th qubit and the identity operation on the remaining qubits.
The unitary operator $U_\theta$, known as a coherent error \cite{ouyang2021avoiding}, can severely damage the qubits, 
especially when quantum control techniques twirl it into a Pauli error basis which cause $U_\theta$ can introduce multiple phase-flip errors on the physical system.

Constant-excitation (CE) codes on $n$ qubits with excitation number $w$ are quantum codes that are supported on computational basis states $|{\bf x}\>$ where ${\bf x}$ are $n$-bit vectors of Hamming weight $w$.
Since any CE code is a subspace of a fixed eigenspace of the operator $Z_1 + \dots +Z_n$, CE codes are left invariant under the action of any coherent error $U_\theta$ for any $\theta$,
and hence completely avoid coherent errors \cite{ouyang2021avoiding}. 
Therefor, CE codes are advantageous to use when coherent errors are the dominant source of errors in quantum systems.

In the construction of \cite{ouyang2021avoiding}, the CE codes are obtained by concatenating any quantum code with the dual-rail code \cite{KLM01}. CE codes obtained through this construction must have an even number of qubits and excitation number equal to half the number of qubits. 
In \cite{hu2021mitigating}, constructions of CSS-type CE codes are also discussed, where the number of qubits can be odd (such as in a [[5,1,2]] CE code). 

Our formalism in this paper further develops the theory of constant excitation codes. 
To obtain constant-excitation code of distance $d$, the key ingredient needed is a classical constant weight code on $n$ bits of distance $d$ that has sufficiently many codewords. 
Now denote $A(n,d,w)$ as the maximal number of classical codewords in a constant weight code on $n$ bits of distance $d$ and weight $w$.
Then our framework can derive CE codes from such constant weight codes whenever
\begin{align}
  A(n,d,w) \ge \sum_{k=0}^{d-1} \binom {n}{k} 
  \label{CE:codeword-numer-req}
\end{align}
for some $w$.
Let $\underline A(n,d,w)$ denote any lower bound to $A(n,d,w)$. Then $\underline A(n,d,w) \ge \sum_{k=0}^{d-1} \binom {n}{k}$ implies that 
$A(n,d,w) \ge \sum_{k=0}^{d-1} \binom {n}{k}$.
Hence we seek lower bounds on $A(n,d,w)$, which we can obtain from the tables in \cite{graham1980lower}. 
Note that any constant weight code has $d\ge 2$, and hence we trivially have $A(n,2,w) = \binom n  w$. Since $\binom 4 2 \ge 1 + 4$, our framework allows us to construct a four-qubit CE code with two excitations.

Using lower bounds on $A(n,d,w)$ from \cite{graham1980lower} where $d=2,4$, we obtain a table of CE codes, specifying what $n_{\rm new}$ and $n_{\rm prev}$ are, where $n_{\rm new}$ and $n_{\rm prev}$ denote how many qubits previous CE codes \cite{ouyang2021avoiding,hu2021mitigating} and our CE codes need respectively.

The table illustrates that our CE codes, for a fixed distance, use fewer qubits than previous CE codes.

\begin{center}
\begin{tabular}{|c|c|c|}
\hline 
Distance ($d$) & $n_{\rm prev}$ & $n_{\rm new}$
\tabularnewline
\hline 
\hline 
2 & 5\cite{hu2021mitigating}, 8 \cite{ouyang2021avoiding} & 4 ($w=2$)
\tabularnewline
\hline 
4 & 20 \cite{ouyang2021avoiding} & 15 ($w=7,8$)\tabularnewline
\hline 
\end{tabular}\medskip{}
\par\end{center}

In the above table, for the constructions using the concatenation framework of \cite{ouyang2021avoiding}, the codes used are the four-qubit error detection code concatenated with the dual-rail code, and a [[10,1,4]] stabilizer code concatenated with the dual-rail code.
While there might be some other error model in which the dual-rail code might outperform the CE codes we construct, this analysis is beyond the scope of our paper.

Using the quantum Gilbert-Varshamov bound for quantum codes we know that $[[n,1,d]]$ asymptotically exist if 
$1-{\rm Ent}_4(d/n)\ge 0$. 
Let $d_{\rm QGV} \approx 0.18929$ be such that $1-{\rm Ent}_4(d_{\rm QGV}/n) = 0$.
Then concatenation of such codes with dual-rail codes gives CE codes with asymptotic relative distance of $d_{\rm QGV}/2 \approx 0.0946$. 
Such CE codes do not attain the quantum Gilbert Varshamov bound.

Next we turn our attention to the asymptotic structure of constant weight codes \cite{graham1980lower}
When the Gilbert bound is applied to constant weight codes, one can obtain \cite[Theorem]{graham1980lower}, which provides the lower bound
\begin{align}
    A(n,2\delta,w) \ge 
    \frac
    {\binom n w}
    {\sum_{i=0}^{\delta-1}\binom w i \binom {n-w} i }
    \ge
    \frac
    {\binom n w}
    {\binom w \delta \binom {n-w}\delta}    .
\end{align}
Hence to construct asymptotic CE codes with $w=n/2$, it suffices to require that 
\begin{align}
    {\binom n {n/2}}
    {\binom {n/2} \delta}^2 
    \ge \sum_{k=0^{2\delta-1}}\binom n k.
\end{align}
or equivalently
\begin{align}
1 - 2{\rm Ent}_2(d/n) \ge 0,
\end{align}
where $d=2\delta$. 
The maximum possible asymptotic value of $d/n$ is approximately $0.11$ which is larger than $0.0946$.
Hence we prove that there are CE codes that asymptotically outperform those obtained using the concatenation construction of \cite{ouyang2021avoiding}.

\subsection{Constructing logical states with linear distance and constant rate}

In building a single logical qubit we used any one non-zero solution in the
kernel of $A$ to identify two disjoint subsets $C_{0}$ and $C_{1}$.
Since the number of rows is $2V_{q}(d_{Z}-1)-1$ and the number of
columns by Gilbert-Varshamov bound satisfies $m\ge q^{n}/V_{q}(d_{Z}-1)$,
from Lemma \ref{lem:main-result} and the asymptotic relation of the volume of the Hamming ball and the $q$-ary entropy function,
we find that the ratio of the number of columns to the number
of rows is asymptotically exponentially large
\begin{equation}
\frac{1}{n}\log_{q}\left(\frac{m}{\#\text{ rows}}\right)\ge1-2{\rm Ent}_{q}\left(\frac{d_{Z}}{n}\right)
\label{eq:number_of_blocks}
\end{equation}
where $d_{Z}/n\in[0,(q-1)/q]$. Here we exploit this to derive roughly $q^{n(1-2{\rm Ent}_{q}(d_{Z}/{n}))}$
logical quantum states, and hence obtaining a linear rate $r\approx(1-2{\rm Ent}_{q}(d_{Z}/{n}))$.
Below we think of $M\propto m/(\#\text{ rows})$. 

We now generalize the construction of two quantum states (a
logical qubit) to more quantum states (a logical qu{\it d}it).
Suppose we identify $M$ subsets $\{C_{0},C_{1},\dots,C_{M-1}\}$
such that $C_{i}\subset C$ for all $i\in\{0,\dots,M-1\}$ and that
the subsets are pairwise disjoint $C_{i}\cap C_{j}=\emptyset$ for
all $i\ne j$. Define the logical qudit as 
\begin{align}
|0_{L}\rangle & = \sum_{\boldsymbol{c}\in C_{0}}\alpha_{\boldsymbol{c}}^{(0)}|\boldsymbol{c}\rangle\: \notag\\
|1_{L}\rangle&=\sum_{\boldsymbol{c}\in C_{1}}\alpha_{\boldsymbol{c}}^{(1)}|\boldsymbol{c}\rangle
\notag\\
&\vdots \notag \\
|(M-1)_{L}\rangle&=\sum_{\boldsymbol{c}\in C_{M-1}}\alpha_{\boldsymbol{c}}^{(M-1)}|\boldsymbol{c}\rangle\:.\label{formM-qudit}
\end{align}
where for every $j=1,\dots, M-1$, we have the normalization
\begin{align}
& \sum_{\boldsymbol{c}\in C_{0}}\left(\alpha_{\boldsymbol{c}}^{(0)}\right)^{2}=\sum_{\boldsymbol{c}\in C_{j}}\left(\alpha_{\boldsymbol{c}}^{(j)}\right)^{2}
\end{align}
Clearly these states are orthonormal. The KL criteria
then writes
\begin{equation}
\Pi\, P\,\Pi=c_{P}\,\Pi\label{eq:KL_generalized}
\end{equation}
where 
$\Pi=|0_{L}\rangle\langle0_{L}|+|1_{L}\rangle\langle1_{L}|+\cdots+|(M-1)_L\rangle\langle (M-1)_L|$ is a projector of (potentially exponentially large) rank $M$.
The orthogonality and non-deformation conditions now write:
\begin{align*}
\langle i_{L}|P|j_{L}\rangle & =c_{P}\:\delta_{i,j}\;,\quad i,j\in\{0,1,2,\dots,M-1\}\:.
\end{align*}
We first check orthogonality. Since the subsets are pairwise disjoint and futhermore $C$ has a minimum distance of $d_X$,
we have $\langle i_{L}|P|j_{L}\rangle=0$ for all $i\ne j$ for and
all diagonal Paulis. Moreover, $\langle i_{L}|P|j_{L}\rangle=0$ if
$1\le\text{wt}_{X}(P)\le d_X-1$, which is inherited from the classical
code's distance $d_X$. We now turn our attention to the non-deformation
condition. As before it is clear that $\langle i_{L}|P|i_{L}\rangle=0$
for all $i\in\{0,1,\dots,M-1\}$ and Paulis with $1\le\text{wt}_{X}(P)\le d_X-1$.
This follows from the distance of the code as before. Therefore, it
is again sufficient to prove the non-deformation condition for diagonal
Paulis. We need to find the subsets $C_{1},C_{2},\dots,C_{M-1}$
such that
\begin{equation}
\langle0_{L}|P|0_{L}\rangle=\langle1_{L}|P|1_{L}\rangle=\dots=\langle (M-1)_{L}|P|(M-1)_{L}\rangle\:.\label{eq:GeneralizedDeformation}
\end{equation}
Our method constructs the sets $C_{0},\dots,C_{M-1}$ such that
they are disjoint and satisfy the foregoing equation. Unlike the qubit case, our proof technique
will not be seeking a solution in the kernel of $A$ anymore, rather
we recursively build the logical quantum states.

An exciting open problem would be to see
an alternative construction that uses the exponentially large kernel
of the matrix $A$.

For every new logical state we need to call the following algorithm once:

\begin{alg}\label{alg:quDit}
{\bf Input:}  A set of $2V_{q}(d_{Z}-1)$
columns of $A$ and a vector ${\bf b}$ of size $\smash{2V_{q}(d_{Z}-1)-1}$. 
\begin{itemize}
\item Check that the augmented homogeneous linear system as described in Section 5 of~\cite{dines1926positive} does not have a row that is all of the same sign.
\item If no row has non-zero entries of the same sign, apply the algorithm of~\cite{dines1926positive} to find a point in the feasible set, i.e., a solution ${\bf x}\ge 0$. End the algorithm. 
\item If there is a row that is all of same sign, augment the set of columns by one and repeat. If all the columns are exhausted, then output fail.
\end{itemize}
{\bf Output:} If succeeded, output a solution ${\bf x}\ge 0$ in the feasible set.
\end{alg}

If recursively successful, then the algorithm can at least be called $m/2V_q(d_Z-1)=q^{n(1-2{\rm Ent}_{q}(d_{Z}/{n}))}$ times. The success is guaranteed if at any step of recursion no row of all the same sign is encountered.

We demonstrate the recursive construction of logical quantum states by first building a qu\textit{trit}
(three logical states). Recall that $A$ has $2V_{q}(d_{Z}-1)-1$
rows whose first row is all ones, we now proceed to find a natural
partitioning of the columns of $A$. To construct a logical qubit,
we use the first $2V_{q}(d_{Z}-1)$ columns of $A$ as follows. Let
the matrix $A'_{1}$ be defined by the first $2V_{q}(d_{z}-1)$ columns
of $A$. Now solve for the solution $\boldsymbol{x_{1}}$ in $A'_{1}\boldsymbol{x=0}$,
where $\boldsymbol{x_{1}}$ has $2V_{q}(d_{Z}-1)$ components. This
matrix equation certainly has a non-trivial kernel because it has
more columns than rows. Just as in the one logical qubit construction
in Alg.~\ref{alg:quBit} we will find an $\boldsymbol{x=x^{+}-x^{-}}$ where
$\boldsymbol{x^{+}}\ge0$ and $\boldsymbol{x}^{-}\ge0$, $\sum_{i}x_{i}^{+}=\sum_{i}x_{i}^{-}$,
and that $supp(\boldsymbol{x^{+}})\cap supp(\boldsymbol{x^{-}})=\emptyset$.
We build the logical qubit exactly as in Eqs.~\eqref{eq:Logical0} and \eqref{eq:Logical1}. 

We then rearrange the columns of $A'_{1}$ according to the supports of
$\boldsymbol{x^+}$ and $\boldsymbol{x^-}$ and write the concatenated matrix $[A_{1} A_{2}]$,
where $A_{1}$ has a set of columns labeled by codewords $c_{s}$ such that $s\in supp(\boldsymbol{x^{+}})$
and $A_{2}$ has columns labeled by codewords $c_{s}$ such that $s\in supp(\boldsymbol{x^{-}})$.
The new equation being satisfied is 
\[
[A_{1}A_{2}]\left[\begin{array}{c}
\boldsymbol{x_{1}^{+}}\\
\boldsymbol{-x_{1}^{-}}
\end{array}\right]=0\;,
\]
where $\boldsymbol{x_{1}^{+}},\boldsymbol{x_{1}^{-}}\ge0$. To build
the third logical state, we select a set of $2V_{q}(d_{Z}-1)$ new
columns out of $A$ and solve 
\[
[A_{2}A_{3}]\left[\begin{array}{c}
\boldsymbol{-x_{1}^{-}}\\
\boldsymbol{x_{2}}
\end{array}\right]=0\;,\quad\boldsymbol{x_{2}\ge0},
\]
where $\boldsymbol{-x_{1}^{-}}$ is treated fixed from the previous
step and we solve for a solution $\boldsymbol{x_{2}}\ge0$. A solution
exists as $A_{3}\boldsymbol{x_{2}}=A_{2}\boldsymbol{x_{1}^{-}}$ is
under-constrained. Although this can be formally thought of as a linear
programming problem, where the objective function is just zero and
a point in the feasible set is sought, we proceed differently and
give an explicit algorithm based on Dines' Annals of Math (1926)~\cite{dines1926positive}.

Once $\boldsymbol{x_{2}}\ge0$ is found this way, we proceed to solve
\[
[A_{3}A_{4}]\left[\begin{array}{c}
\boldsymbol{x_{2}}\\
\boldsymbol{-x_{3}}
\end{array}\right]=0\;,\quad\text{\ensuremath{\boldsymbol{x_{3}}}\ensuremath{\ge0},}
\]
where now $\boldsymbol{x_{3}}$ is the unknown. Just as before this
amounts to solving $A_{4}\boldsymbol{x_{3}}=-A_{3}\boldsymbol{x_{2}}$
where the right-hand side is a known vector. Continuing this way we
can build $2\le M\le q^{n(1-2{\rm Ent}_{q}(\frac{d_{Z}}{n}))}$ logical states.

The results of this section prove the following main theorem of part 1 of this work:
\begin{thm}\label{Thm:main1}
Take a random classical code $C$ of length $n$ on a $q$-ary alphabet with
a minimum distance of $d_X$. 
Then almost surely with probability equal to 1 as $n$ goes to infinity, we have $|C|\ge 2V_{q}(d_Z-1)$.
Moreover, using Alg.~\ref{alg:quBit} and multiple calls to Alg.~\ref{alg:quDit} explicitly, we will with probability 1 as $n$ goes to infinity derive the $M$ quantum logical states in \eqref{formM-qudit} with bit- and phase-flip distances of $d_X$ and $d_Z$ respectively. The overall distance is $\min(d_X,d_Z)$. 
\end{thm}
When the $A$ matrix is wide (under-constrained) Alg.~\ref{alg:quBit} always constructs two logical codewords (a logical qubit). The third and subsequent logical codewords are found recursively by calling Alg.~\ref{alg:quDit} multiple times.
Alg.~\ref{alg:quDit} succeeds with high probability and almost surely over random codes. In the rare case that the set of linear constraints have rows that are all non-zero and are of the same sign, the recursion becomes infeasible and algorithm halts. As shown in~\cite{dines1926positive} this is the only way that the algorithm can fail. 
We define a random classical code as one whose codewords have entries over $\Sigma=\{0,1,2,\dots,q-1\}$ such that each entry is independently and randomly drawn from the uniform distribution over $\Sigma$. 

\begin{lem}
$M$ quantum logical states can be constructed as long as in each step of the recursion, no row other than the first has entries that are all non-zero and with the same sign.  When  
$C$ is a random classical code, then with high probability, Alg.~\ref{alg:quDit} succeeds in providing a quantum code with constant rate.
\end{lem}
\begin{proof}
The first part of the Lemma follows from the proof of Dines~\cite{dines1926positive} applied to every step of the recursion. To prove the second part, first recall that at each step we are solving a linear system with $\rho\equiv 2V_q(d_Z-1)$ columns and $\rho-1$ rows.
To find positive solutions of this linear system, we employ Dines algorithm which is itself recursive. Hence we will prove that at every step of Dines recursive algorithm, it fails with low probability. We prove this by induction, first starting with the base case.

Recall that $\omega\equiv \exp(2\pi i/q)$. For every diagonal $P$ of weight at least one, $\<c|P|c\>$ is a random variable that takes the values
$\{1,\omega, \dots \omega^{q-1}\}$ with a uniform probability. From Eq.~\eqref{eq:cosine-sine} we know that each $\<c|P|c\>$ is equal to a $\omega^j$ for some $j\in\{0,1,\dots,q-1\}$. That is, ${\rm Re}(\<c|P|c\>)$ and ${\rm Im}(\<c|P|c\>)$ are random variables that take values in $[-1,1]$. Since the code words are random and uniformly distributed over the symbols, by symmetry the probability of an entry having a positive (or negative) sign is a half when $q$ is even and is at most $2/3$ when $q$ is odd.
Moreover the entries are independent. Hence
the probability that a given row has all the same sign is $(3/2)^{-\rho+1}$. And by a union bound, the probability that any of the $\rho-1$ rows have entries whose all entries have the same sign is $O(\rho\,(3/2)^{-\rho})$.

Now we prove the induction step. We take $q$ to be even for now, which  ensure that $a_{ij}$ are symmetric random variables with mean zero. Note that Dines algorithm takes a matrix with matrix elements $a_{ij}$, and constructs a new matrix $a_{r,ij} = a_{1i}a_{rj} - a_{1j}a_{ri}$.
In the new matrix, the indices $i$ and $j$ belong to disjoint sets $I$ and $J$. The number of columns in the new matrix is $|I||J|$. For the induction hypothesis, we assume that the matrix elements $a_{ij}$ are identical and symmetric random variables with zero mean, which are furthermore independent with respect to the column index $j$. We will show that $a_{r,ij}$ will then also be a symmetric random variable with zero mean and also furthermore be independent with respect to the new column indices $ij$.

We now show that the $a_{r,ij}$ has zero mean. For that we see $\mathbb E(a_{r,ij}) 
= \mathbb E(a_{1i}a_{rj}) - \mathbb E(a_{1j}a_{ri})$ because the expectation is linear. Next the independence of i and j imply that $\mathbb E(a_{1i}a_{rj}) = \mathbb E(a_{1i})\mathbb E(a_{rj})$ and $\mathbb E(a_{1j}a_{ri})=\mathbb E(a_{1j})\mathbb E(a_{ri})$. Substituting this shows that 
$\mathbb E(a_{r,ij}) = 0$ because $a_{ij}$ are independent random variables with mean zero.

We now note that the random variables $a_{r,ij}$ are independent with respect to the column labels $ij$. This follows readily from the independence of $a_{ij}$ with respect to $j$. For instance, treat $i$ to be fixed and consider $a_{r,ij}$ and $a_{r,ij'}$ for $j\ne j'$.

We next show that the random variables $a_{r,ij}$ are symmetric. For this, we use the fact that the product of non-degenerate symmetric random variables is a symmetric random variable \cite{hamedani1985product}.

In the base case, we have shown that $a_{ij}$ satisfies the induction hypothesis with high probability.
We have just shown that $a_{r,ij}$ is symmetric, independent, and has zero mean with respect to $ij$. This proves that its entries have equal probability of being positive or negative.

It remains to bound the probabilities of failure of our algorithm under our recursion and Dines recursion. 
With high probability, the product $|I||J|$ is going to be greater than the number of columns in $a_{ij}$. As shrinking the number of columns will only be due to at most one entry of a different sign, this happens with very low probability. The probability of this happening is at most $2 \rho (1/2)^{\rho-1}$ for large $\rho$. The probability of a single row with all the same sign is at most $2 (1/2)^{\rho}$. Hence the total probability of a given row being pathological is at most $2(1/2)^{\rho} + 2\rho (1/2)^{\rho-1} = O(\rho (1/2)^\rho)$.
The probability of a matrix at any step of Dines recursion to have a pathological row is therefore at most $O(\rho^2 (1/2)^\rho))$. Since our algorithm has $O(m/\rho)$ steps, our algorithm's failure probability is at most 
$O(m\rho (1/2)^\rho)$. The number of codewords in our $q$-random code is $m \le q^n$, and therefore $m 2^{-\rho}$ is at most $(q^n 2^{-2^{c n}})$ for some constant $c \in [0,1]$, our algorithm will succeed with overwhelming probability. Although this proof is specialized for the case where $q$ is even, a similar argument will work for $q$ odd.
\end{proof}

In the next section, 
we find that even if this happens we can always construct {\it Approximate} Quantum Error Correcting Codes (AQECC)~\cite{LNCY97} with linear distance and constant rate, provided that, the underlying classical code also has a linear distance and a sufficiently large rate.

\subsection{AQECCs with designed rates}

In the unlikely case, where  building logical quantum states using Alg.~\ref{alg:quDit}  of the previous section fails, we can always build an AQECC as we now show. We introduce an algorithm that produces a quantum code with $M$ logical codewords satisfying the KL criteria approximately. Here $M$ can be strictly larger than 2 at the expense of an approximation error, which is equal to the infidelity of quantum code.

Suppose $C_1,\dots, C_{M/2}$ are  disjoint subsets of the classical code $C$ whose minimum distance is $d_X$. Hence each $C_j$, $j\in[M/2]$ inherits the distance $d_X$, where we take $M$ to be even for simplicity.
Suppose that for every $j\in[M/2]$, it holds that
\begin{align}
    |C_j| \ge 2 V_q(d_Z-1)\;.
\end{align}
For each classical code $C_j$, we use Alg.~\ref{alg:quBit} to construct a corresponding matrix $A_j$ from which we derive logical codewords $|(2j)_L\> , |(2j-1)_L\>$ that satisfy the KL criteria for quantum codes with a bit-flip and phase-flip distance of $d_X$ and $d_Z$ respectively.

It is clear from our construction that for every diagonal Pauli $P$ of weight at most $d_Z-1$, 
\begin{align}
    \< (2j)_L   |\,P\, | (2j)_L\> 
    = \< (2j-1)_L |\,P\,| (2j-1)_L\> = \gamma_{j,P} \label{eq:balanced-sandwiches}
\end{align}
where $\gamma_{j,P}$ is a complex number of norm at most one.  
Now for each $j\in[M/2]$, let ${ \Gamma_{j}} = (\gamma_{j,P})$ be a row vector of length $V_q(d_Z-1)$ with components that correspond to diagonal Paulis of weight at most $d_Z-1$.
Then it follows that each $\Gamma_{j}$ lies in a $|V_q(d_Z-1)|$-dimensional complex unit ball corresponding to a hyper-cube of length two centered at the origin with respect to the infinity norm.  
 
Suppose $ \max_{j,k \in [M/2]} \| \Gamma_j  -  \Gamma_k \|_\infty =  \delta$. Then because of Eq.~\eqref{eq:balanced-sandwiches} we have 
\begin{align}
   \max_{j,k=1,\dots,M} \max_P \left| \<j_L|\,P\,|j_L\> - \<k_L|\,P\,|k_L\> \right|= \delta. \label{lem4eq}
\end{align} 
It remains to find a suitable upper bound for $\delta$. 
To relate the error $\delta$ in satisfying the non-deformation to the size of the code $|C|$ and the number of logical qubits we can construct, we rely on the following fact.
\begin{fact*}
Consider the complex hyper-cube of side length two in $N$ dimensions. Let $x$ be the number of points distributed randomly inside it. Then there exists a ball of radius $\delta$ in the infinity norm that contains at least $x (\delta/2)^N$ points in expectation.
\end{fact*}
The number of points inside the unit hyper-cube is $\smash{x=\lfloor |C|/(2V_q( d_Z-1))\rfloor}$. The radius of the ball is $\delta$. The dimension of the hyper-cube is $V_q(d_Z-1)$.
Hence from the above fact, we have that the expected number of logical codeword pairs is $\mathbb{E}M=2x (\delta/2)^N$ which writes
\begin{align}\label{eq:ExpectedLogical}
    \mathbb{E}M \ge 2\left\lfloor \frac{|C|}{2V_q( d_Z-1)}\right\rfloor 
    \left(\frac \delta  2 \right)^{V_q(d_Z-1)}.
\end{align}
 
The infidelity $\epsilon$ defined by one minus the worst case entanglement fidelity of the quantum code can be upper bounded as shown  
in~\cite{ouyang2014permutation}, to be
\begin{align}
    \epsilon \le 
    O\left( 
     \delta\; V_q^4(d_Z-1)
    \right),
\end{align}
when the noisy quantum channel introduces bit-flip and phase-flip weights at most $d_X-1$ and $d_Z-1$ uniformly at random.

\subsection{Illustrations}
In this section we illustrate our framework through a series of examples. It is noteworthy that the first example is a one-to-one correspondence between the celebrated classical and quantum results of Hamming's and Steane's respectively.

\subsubsection{{\it C} as a [7,4,3] Hamming code gives the Steane code}
\label{sec:steane-code}

In classical coding theory, we use the notation $(n,m,d)$ to denote a binary code with codewords of length $n$ that has $m$ codewords, a distance of $d$. We use $[n,\log_2 m,d]$ to denote a binary code with codewords of length $n$ that has $m$ codewords, a distance of $d$, and is furthermore a linear code.
Consider the case when $C$ is generated from the codewords 
1000110,    
0100101,
0010011, and 
0001111.
This classical code is the celebrated [7,4,3] Hamming code, and has been used previously by Steane to obtain the [[7,1,3]] Steane code.
Applying our framework to $C$, we get the unique solution
\begin{align}
    |0_L\>
    =&
    \frac{1}{\sqrt{8}}\left(
    |0,0,0,0,0,0,0\>
+|0,0,0,1,1,1,1\> \right.\notag\\&\left.\quad 
+|0,1,1,0,1,1,0\>
+|0,1,1,1,0,0,1\>\right.\notag\\
&\left.+|1,0,1,0,1,0,1\>
+|1,0,1,1,0,1,0\>
\right.\notag\\&\left.\quad 
+|1,1,0,0,0,1,1\>
+|1,1,0,1,1,0,0\>\right)\\
|1_L\> = &
\frac{1}{\sqrt{8}}
\left(|0,0,1,0,0,1,1\>
+|0,0,1,1,1,0,0\>
\right.\notag\\&\left.\quad 
+|0,1,0,0,1,0,1\>
+|0,1,0,1,0,1,0\>\right.\notag\\
&\left.+|1,0,0,0,1,1,0\>
+|1,0,0,1,0,0,1\>
\right.\notag\\&\left.\quad 
+|1,1,1,0,0,0,0\>
+|1,1,1,1,1,1,1\>\right).
\end{align}
This is in fact equivalent to the Steane code.

\subsubsection{{\it C} as a nonlinear cyclic code}
\label{sec:nonlinear-cyclic}
We now give an example of a quantum code constructed from a non-linear classical code.  Take the nonlinear $(4,8,2)$ code with codewords
that are cyclic permutations of 
$0001$ and $1110$.
The corresponding kernel of the matrix $A$ has dimension 3, and one solution to this gives an error detecting quantum code with logical codewords
\begin{align}
    |0_L\> &= \frac{1}{\sqrt 2}(
    |0,0,0,1\>
    +|1,1,1,0\>)\\
    |1_L\> &= \frac{1}{\sqrt 2}(
    |0,0,1,0\>
    +|1,1,0,1\>).
\end{align}
In fact, we can have the additional logical codewords
\begin{align}
    |2_L\> &= \frac{1}{\sqrt 2}(
    |0,1,0,0\>
    +|1,0,1,1\>)\\
    |3_L\> &= \frac{1}{\sqrt 2}(
    |1,0,0,0\>
    +|0,1,1,1\>).
\end{align}
This gives a quantum code of dimension 4 and a minimum distance of 2.
This quantum code is also an example of a CWS code. This is because for every $j=0,1,2,3$, there is a Pauli operator that takes 
the stabilizer state 
$\frac{1}{\sqrt 2}(
    |0,0,0,0\>
    +|1,1,1,1\>)$ to
$|j_L\>$.
Since this quantum code is CWS, quantum error correction can proceed using formalism developed for CWS codes~\cite{CSSZ08}. 
This quantum code also has some other attractive properties. First, by inducing cyclic shifts in the underlying qubits, we can move from one logical codeword to another. 
Second, this quantum code is stabilized by $X^{\otimes 4}$ and is affected uniformly by $Z^{\otimes 4}$. Together, this implies that the quantum code is invariant under transversal $X$ and $Y$ and $Z$ operations.

\subsubsection{Permutation-invariant quantum codes}\label{subsec:permInva}
Our quantum code construction formalism can also extend to quantum codes with logical codewords that are supported on non-product basis states. 
One example of such codes are permutation-invariant quantum codes, which are invariant under any permutation of the underlying particles. Permutation-invariant quantum codes have been explicitly constructed using a variety of different techniques~\cite{Rus00,PoR04,ouyang2014permutation,ouyang2015permutation,OUYANG201743,ouyang2019permutation}.
Recently, permutation-invariant quantum codes have been considered for applications such as for quantum storage~\cite{ouyang2019mems}, or for robust quantum meterology~\cite{ouyang2019robust}, and they can also be prepared in physically realistic scenarios~\cite{init-picode-2019-PRA}.

When permutation-invariant quantum codes are constructed on $n$ qubits, they must be superpositions over Dicke states
\begin{align}
|D^n_w\> = \frac{1}{\sqrt{\binom n w}} \sum_{\substack{ x_1,\dots, x_N \in \{0,1\} \\ x_1+\dots + x_n = w }} |x_1\>\otimes \dots \otimes |x_n\>.\label{eq:dicke-defi}
\end{align}
Here $w$ is the weight of the Dicke state, and counts the Hamming weights of its constituent computation basis states' labels.
The Dicke states for qubit states are labeled by only their weights, of which there are only $n+1$ possibilities. 
For our quantum code construction, we can choose the logical states to be supported on $|D^n_{w_1}\> ,\dots ,|D^n_{w_m}\>$ where 
$w_{j+1}-w_j \ge d$ for any $j=1,\dots, m-1$,
and $d$ is the desired minimum distance of the quantum code.

When a quantum code is permutation-invariant, we only need to consider equivalence classes of Pauli operators up to a permutation. 
For Dicke states, 
$\<D^n_w|P|D^n_w\>$ are not necessarily zero even when the Pauli $P$ is non-diagonal~\cite{ouyang2019robust}, we need to count the number of all Paulis of weight at most $d-1$ up to a permutation.
The number of unique qubit-Paulis up to a permutation having a weight of at most $w$ is equal to the number of ways to order a $w$-tuple in $\{1,2,3\}^w$ in a non-decreasing sequence, and this number is just $\binom {n+w-1}{w}$. 
Hence to total number of Paulis that we need to consider for the non-deformation conditions is  at most
$\sum_{w=0}^{d-1}3^w.$

Now we consider a variation of the $A$-matrix from Alg.~\ref{alg:quBit} with the matrix elements
\begin{align}
    \<D^n_{w_j}  | P | D^n_{w_j} \>,
\end{align}
where $P$ labels the rows and $j$ labels the columns. 
 
From this, we can get permutation-invariant quantum codes with a minimum distance of $d$ whenever
\begin{align}
    (\lfloor n / d \rfloor +1) \ge 1+\sum_{w=0}^{d-1}3^w .
\end{align}
For instance, when $d=3$ this inequality becomes
\begin{align}
    \lfloor n/3 \rfloor \ge  (1+13),
\end{align}
and this formalism show that we can get gives permutation-invariant codes with a distance of $d=3$ when $n\ge 42$. This bound is however loose, because there are permutation-invariant quantum codes with $d=3$ on 9 qubits~\cite{Rus00,ouyang2014permutation}, and even on 7 qubits~\cite{PoR04}.
This suggests that rather than using loose bounds on the nullity of the $A$ matrix in Alg.~\ref{alg:quBit}, we need to exploit additional structure about the kernel of $A$ to realize the full potential of our formalism.
 
\subsubsection{Remarks on optimality}

If we take the metric of optimality to be maximization of the distance $d$ for fixed length $n$ and number of encoded qubits $k$, then our construction is not optimum. This is because the five-qubit code is the unique [[5,1,3]] code, and our framework does not encompass it. To see this, note that the classical codewords on which the [[5,1,3]] code is supported on comprise of a distance-1 classical code.

We demonstrate examples of optimal quantum codes using our framework which use the Hamming code and a nonlinear cyclic code in Sec.~\ref{sec:steane-code} and Sec.~\ref{sec:nonlinear-cyclic} respectively.  

\subsection{Recovery procedure}
If a quantum code satisfies the KL criteria \cite{KnL97} exactly, we can use the KL procedure  (Ref.~\cite{KnL97}) to construct explicit recovery maps for the quantum code. We briefly review this recovery procedure. 

Given an $M$-dimensional quantum code, the correctible subspaces reside within a space spanned by
$W_{1}|k_{L}\rangle
, \dots ,
W_{w}|k_{L}\rangle
$, where $k\in \{0,1,\dots, M-1\}$ and $W_j$ are matrices that we call correctible errors. 
Using any Gram-Schmidt procedure on 
$
W_{1}|k_{L}\rangle
, \dots ,
W_{w}|k_{L}\rangle
$
for every $k\in \{0,1,\dots, M-1\}$,
we obtain a set of orthonormal states
\[
|F_{k,1}\rangle,\dots,|F_{k,r}\rangle.
\]
Here $r$ is the dimension of the space spanned by $W_{1}|k_{L}\rangle,\dots,W_{w}|k_{L}\rangle$, and is identical for every $k$.
For any such Gram-Schmidt decomposition, and for any $j\in[r]$, we can form the projectors
\begin{align}
    P_{j} = \sum_{k=0}^{M-1} |F_{k,j}\>\<F_{k,j}|.
\end{align}
From the KL criteria, these projectors are pair-wise orthogonal, and we can derive the corresponding unitary operators $U_j$ so that the recovery operation with Kraus operators $U_j P_j$ will correct all errors spanned by the elements of the correctible set $\Omega$.
These Kraus operators correspond
to the operational procedure of first performing projective measurements given by the projectors $\{P_1,\dots, P_r\}$, obtaining an error syndrome $j$ that corresponds to the projector $P_j$, and applying a conditional unitary operation $U_j$ to complete the recovery procedure.

This KL recovery procedure does not guarantee a priori efficient implementations of these recovery maps. Indeed, the complexity of these projective measurements and conditional unitary operations when decomposed into individual experimentally accessible operations or quantum-gate sequences remains to be elucidated. We believe that a constructing a systematic and efficient decoding scheme that is compatible with our quantum coding framework is an important milestone to be achieved in bringing our quantum codes closer toward implementation.

\section{Part 2: Linear distance codes in ground space of local Hamiltonians}

\subsection{Local Hamiltonian and its ground space}

Let us consider a spin chain of length
$n$ with open boundary conditions and the local Hilbert space dimension of $2s+1$, where $s\ge1$
is a positive integer. We take a representation in which $|j\rangle$
denotes the $s_{z}=j$ state of a spin-$s$ particle:
\[
\hat{S}_{z}|j\rangle=j\;|j\rangle,\qquad j\in\Sigma\quad.
\]

The local Hamiltonian whose ground space will be shown to contain a nontrivial quantum
error correcting code is
\begin{equation}
H_{n}=H_{n}^{J}+H_{n}^{s}\label{eq:H}
\end{equation}
where $H_{n}^{J}=J\sum_{k=1}^{n}\left(|0\rangle\langle0|\right)_{k}$.
Recall that the Hamiltonian $H_{n}^{s}$, is defined by
\begin{equation}
H_{n}^{s}=\sum_{k=1}^{n-1}\left\{ \sum_{m=-s}^{s}P_{k,k+1}^{m}+\sum_{m=1}^{s}Q_{k,k+1}^{m}\right\} \;,\label{eq:Hn_s-2}
\end{equation}
and the local terms are projectors acting on two neighboring spins $k,k+1$ defined by 
\begin{align}
P^{m}&=|0\leftrightarrow m\rangle\langle0\leftrightarrow m|
\notag\\
Q^{m}&=|00\leftrightarrow\pm m\rangle\langle00\leftrightarrow\pm m|;\label{eq:hs_Proj}
\end{align}
where
\begin{eqnarray}
|0\leftrightarrow m\rangle & \equiv & \frac{1}{\sqrt{2}}\left[|0,m\rangle-|m,0\rangle\right]\label{eq:Local_move}\\
|00\leftrightarrow\pm m\rangle & \equiv & \frac{1}{\sqrt{2}}\left[|0,0\rangle-|m,-m\rangle\right],\label{eq:local_int}
\end{eqnarray}
and we denoted by $|j,k\rangle$ the spin state $|s_{k}^{z}=j,s_{k+1}^{z}=k\rangle$.
There are $3s$ local projectors as $P_{k,k+1}^{0}$ automatically
vanishes. See Fig.~\ref{Fig:Hamiltonian}.

In the simplest form $s=1$, and we have
\begin{align*}
H_{n}^{1}=&\frac{1}{2}\sum_{k=1}^{n-1}\left\{ |0\leftrightarrow1\rangle\langle0\leftrightarrow1|\:+\:|0\leftrightarrow-1\rangle\langle0\leftrightarrow-1|\:\right.\\
&\left.\quad\quad+\:|00\leftrightarrow\pm1\rangle\langle00\leftrightarrow\pm1|\right\} \;,
\end{align*}
where $|1\leftrightarrow0\rangle\propto|0,1\rangle-|1,0\rangle$,
$|-1\leftrightarrow0\rangle\propto|0,-1\rangle-|-1,0\rangle$, and
$|\pm1\leftrightarrow00\rangle\propto|0,0\rangle-|1,-1\rangle$. Below
we will be mostly interested in $s>1$,
because the case of $s=1$ has a ground space of dimension too small for us to construct non-trivial QEC codes, and the case of $s>1$ offers a ground space of large enough dimension for us to construct non-trivial QEC codes.
\begin{figure*}
    \centering
    \includegraphics[scale=0.50]{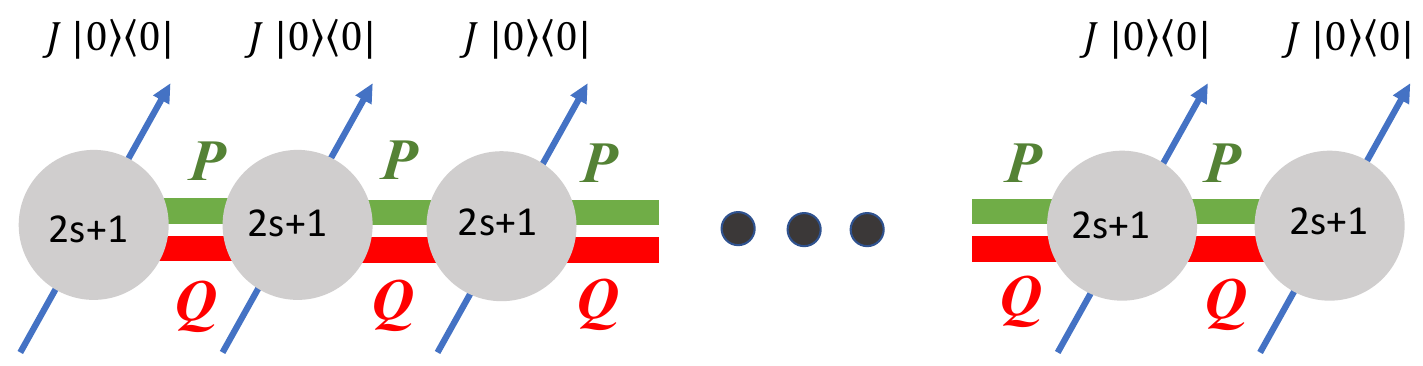}
    \caption{The new local integer spin-$s$ Hamiltonian $H_n$.}
    \label{Fig:Hamiltonian}
\end{figure*}
\begin{lem}
\label{lem:FF_Sum_Lemma} Suppose $H_{1}\ge0$ and $H_{2}\ge0$, and
$H_{1}+H_{2}$ is a frustration free (FF) Hamiltonian with zero energy
ground state. Then the ground space of $H_{1}+H_{2}$ coincides with
the intersection of the ground spaces of $H_{1}$ and $H_{2}$.
\end{lem}
\begin{proof}
Any state in the intersection of the kernels of $H_{1}$ and $H_{2}$
automatically vanishes on $H_{1}+H_{2}$. Conversely, a state $|\psi\rangle$
that is in the kernel of the sum $H_{1}+H_{2}$ obeys 
$\smash{\langle\psi|(H_{1}+H_{2})|\psi\rangle=0}$.
Since each summand is a positive operator, so is their sum, and for $|\psi\rangle$ to be a zero energy ground state of $H_{1}+H_{2}$ it has to vanish on each summand $\langle\psi|H_{i}|\psi\rangle=0$ for $i=1,2$. Therefore
$|\psi\rangle$ is a FF ground state of each $H_{i}$ as well.
\end{proof}
It is clear that
\begin{equation}
\text{spec}(H_{n}^{J})=J\{0,1,2,\dots,n\}\label{eq:SpecH_J}
\end{equation}
whose gap we denote by $\Delta(H_{n}^{J})=J$. The kernel of $H^J_n$ is the
span of all product states $|{\bf t}\rangle$ of weight $n$, where
${\bf t}\in\Sigma_{*}^{n}$; note that the letter $0$ is excluded in these strings.

We will obtain the ground space of $H_{n}$ by taking the intersection
of the ground space of $H_{n}^{J}$ with $H_{n}^{s}$. We proceed
to analytically derive the ground space of $H_{n}^{s}$ after preliminary definitions.
From now we assume that $m$ is a positive integer unless stated otherwise.

Since $H_{n}^{s}$ is free of the sign problem (i.e.,~stoquastic),
the local projectors define an effective Markov chain, which have
the correspondence shown in Table \ref{table:moves}.


We say two strings ${\bf t},{\bf z}$ are equivalent, denoted by ${\bf t}\sim{\bf z}$
if ${\bf z}$ can be reached from ${\bf t}$ by applying a sequence
of the local moves defined in the table. We define a set of equivalence
classes as follows. Let $k$ denote the number of nonzero letters
in a product state (Eq.~\eqref{eq:Equiv_colored}). Using a consecutive
set of the local moves stated above, we can take any state $\dots m0\dots0(-m)\dots\rightarrow\dots0\dots0m(-m)0\dots0\dots\rightarrow\dots0000\dots$
. We then move all the zeros to the rightmost end and ensure that
all strings are of the form
\begin{eqnarray}
c_{x_{1},\dots,x_{k}} & = & x_{1}\dots x_{k}\underbrace{0\dots0}_{n-k}\label{eq:Equiv_colored}
\end{eqnarray}
where $x_{i}\in\Sigma_{*}$. By assumption the string $x_{1}\dots x_{k}$
cannot be further reduced and $n-k$ is the maximum number of zeros.
Then it follows that if $x_{i}=m$ then it must \textit{not} have
to its immediate right an $x_{i+1}=-m$ for otherwise the annihilation
rule $(m,-m)\rightarrow00$ would further reduce it. 
\begin{lem}
Any string ${\bf t}\in\Sigma^{n}$ is equivalent to one and only one
$c_{x_{1},\dots,x_{k}}$ (Eq.~\eqref{eq:Equiv_colored}).
\end{lem}
\begin{proof}
By applying the local moves in the table above to any string ${\bf t}$,
one can make sure that there are no substrings $m(-m)$ or $m0\cdots0(-m)$
for any $m\in[s]$, 
where $[s]=\{1,2,\dots,s\}$. Now suppose we apply these moves to ${\bf t}$
as much as possible to bring it as close as possible to the state
of all zeros. Then if ${\bf t}$ contains a single $m$, then the
first non-zero letter to its right cannot be $-m$. Similarly if ${\bf t}$
contains at least one $-m$ then the first non-zero letter to its
left must not be $m$. By applying $0(-m)\rightarrow(-m)0$
and $0m\rightarrow m0$ we move all the zeros to the right to obtain
a string of the form given by Eq.~\eqref{eq:Equiv_colored}. To prove
that the set of all strings equivalent to $c_{x_{1},\dots,x_{k}}$
is indeed an equivalence class, we need to prove that the classes
are distinct. It is clear that any string is equivalent to itself
(reflexive). If ${\bf x}\sim{\bf y}$ and ${\bf y}\sim c_{x_{1},\dots,x_{k}}$
then ${\bf x}\sim c_{x_{1},\dots,x_{k}}$ (transitive). Lastly if
${\bf x}\sim{\bf y}$, then ${\bf y}\sim{\bf x}$ because of the reversibility
of the local moves (symmetric). Therefore indeed the set of strings
equivalent to $c_{x_{1},\dots,x_{k}}$ form an equivalence class and
it is an elementary fact that equivalence classes are distinct and
partition the state space (i.e., the set of all strings) into disjoint subsets.
\end{proof}
\begin{lem}
\label{lem:GS_Hn}The uniform superposition of all strings in an equivalence
class (i.e., equivalent to the irreducible string in Eq.~\eqref{eq:Equiv_colored})
 is a (frustration free) zero energy ground state of $H^s_n$.
\end{lem}
\begin{proof}
The Hamiltonian $H_{n}^{s}$ in Eq.~\eqref{eq:Hn_s-2} is a sum of local
projectors (Eq.~\eqref{eq:hs_Proj}). If a ground state $\psi$ vanishes
on each local projector, then for any $m\in[s]$ it must obey $\langle\psi|0m\rangle=\langle\psi|m0\rangle$,
$\langle\psi|0(-m)\rangle=\langle\psi|(-m)0\rangle$ and $\langle\psi|00\rangle=\langle\psi|m(-m)\rangle$.
It follows that $\psi$ has the same amplitude on a pair of equivalent
strings ${\bf s}\sim{\bf t}$, which means $\langle\psi|{\bf s}\rangle=\langle\psi|{\bf t}\rangle$.
It follows that the ground subspace of $H_{n}^{s}$ is frustration
free and is spanned by the pairwise orthogonal states
\begin{equation}
|c_{{\bf x_{k}}}\rangle\propto\sum_{{\bf s}\sim c_{{\bf x_{k}}}}|{\bf s}\rangle\label{eq:AllGroundStates}
\end{equation}
where to simplify the notation, we denoted ${\bf x_{k}}=(x_{1},\dots,x_{k})$.
Clearly each distinct ${\bf x_{k}}$ results in a distinct ground
state $|c_{{\bf x_{k}}}\rangle$. 
\end{proof}
The ground states can be highly entangled. However, among
the many ground states there is a substantial subset that are all
product states, i.e., $k=n$. We will use these to construct quantum
error correcting codes. Before doing so let us answer: How many product
state ground states are there?

Let $T_{n}$ be the set of all \textit{allowed} $2s$-ary strings
${\bf t}=t_{1}t_{2}\dots t_{n}$ of length $n$ defined by
\begin{equation}
T_{n}\equiv\left\{ {\bf t}\in\Sigma_{*}^{n}\:|\:\text{if}\;t_{j}=m,m\in[s],\text{ then }t_{j+1}\ne-m\right\} .\label{eq:Tn}
\end{equation}
Let $|T_{n}|$ be the size of this set. Since $T_{0}=\emptyset$ and
$T_{1}=\Sigma_{*}$, we have that $|T_{0}|=1$ and $|T_{1}|=2s$. 
\begin{lem}
$|T_{n+2}|=2s|T_{n+1}|-s|T_{n}|$ , with $|T_{0}|=1$ and $|T_{1}|=2s$.
We have
\begin{equation}
|T_{n}|=\frac{
s^{n}\left\{ \left(1+\bar s\right)^{n+1}-\left(1-\bar s\right)^{n+1}\right\} 
}{2\sqrt{1-1/s}}
,\label{eq:Tn_size}
\end{equation}
where $\bar s = \sqrt{1-1/s}$.
Asymptotically, for $n \gg 1$, it holds that
\begin{equation}
|T_{n}|\approx\frac{1+\sqrt{1-1/s}}{2\sqrt{1-1/s}}\left[s(1+\sqrt{1-1/s})\right]^{n}.
\label{eq:TnAsymp}
\end{equation}
\end{lem}
\begin{proof}
We prove this by induction, any ${\bf t}\in T_{n}$ is ${\bf t}=t_{1}\dots t_{n}$,
where $t_{j}\in\Sigma_{*}$. Since $t_{1}$ can be either a $m$ or
$-m$ for some $m\in[s]$, ${\bf t}$ is either ${\bf t}=(-m){\bf t}'$,
where the string ${\bf t}'\in T_{n-1}$ or ${\bf t}=m{\bf t}'$ where
${\bf t}'$ denotes the subset of strings in $T_{n-1}$ that do not
start with the letter $-m$, i.e., $t'_{1}\ne-m$. The number of strings
${\bf t}=(-m){\bf t}'$ with ${\bf t}'\in T_{n-1}$ is clearly $s|T_{n-1}|$.
Now the set of all string ${\bf t}=m{\bf t}'$ with $t'_{1}\ne-m$
coincides with the set that excludes the strings ${\bf t}=m(-m){\bf t}''$
where ${\bf t}''\in T_{n-2}$. Since $m$ takes on $s$ different
values, we have that the number of strings ${\bf t}=m{\bf t}'$ with
$t'_{1}\ne(-m)$ is $s\left(|T_{n-1}|-|T_{n-2}|\right)$. 

The total size of the set is then $|T_{n}|=2s|T_{n-1}|-s|T_{n-2}|$,
which is a linear recursion of second order with initial conditions
$|T_{0}|=1$ and $|T_{1}|=2s$. Shifting the indices to reflect $|T_{0}|$ and $|T_{2}|$
as the starting values, we have
\[
|T_{n+2}|=2s|T_{n+1}|-s|T_{n}|.
\]
The solution is elementary and of the form $|T_{n}|=Ar_{+}^{n}+Br_{-}^{n}$
, where the two roots $r_{\pm}$ are $r_{\pm}=s(1\pm\sqrt{1-1/s})$
and $A=\frac{1+\sqrt{1-1/s}}{2\sqrt{1-1/s}}$, $B=\frac{-1+\sqrt{1-1/s}}{2\sqrt{1-1/s}}$
are obtained from the initial conditions $|T_{0}|=1$ and $|T_{1}|=2s$.
This proves Eq.~\eqref{eq:Tn_size} and the observation $\smash{(1-\sqrt{1-1/s})<1}$
proves the asymptotic formula Eq.~\eqref{eq:TnAsymp}.
\end{proof}
In the limit we have  $\lim_{s\rightarrow1}|T_{n}|=n+1$, which is the number of distinct
product ground states $|(-)_{1}\dots(-)_{p}(+)_{p+1}\dots(+)_{n}\rangle$
where $p\in\{0,1,\dots,n\}$ with $p=0$ corresponding to $|++\dots+\rangle$.
\begin{cor}
The dimension of the kernel of $H_{n}^{s}$ is $\sum_{k=0}^{n}|T_{k}|=\frac{\left(-2+r_{+}^{n+1}+r_{-}^{n+1}\right)}{2(s-1)}$,
where $r_{\pm}\equiv s(1\pm\sqrt{1-1/s})$. Asymptotically we have
$\dim(\ker(H_{n}^{s}))\approx r_{+}^{n+1}/[2(s-1)]$.
\end{cor}
\begin{proof}
The total number of equivalent classes is the dimension of the kernel.
For each $k\in[n]$ there are $|T_{k}|$ equivalent classes and we
have
\begin{align*}
\dim(\ker(H_{n}^{s})) & =\sum_{k=0}^{n}|T_{k}|=\frac{r_{+}^{n+1}+r_{-}^{n+1}-2}{2(s-1)}.
\end{align*}
\end{proof}
\begin{rem}
The fraction of product state ground states is a constant independent
of $n$
\begin{align*}
\frac{|T_{n}|}{\sum_{k=0}^{n}|T_{k}|}
&=\sqrt{\frac{s-1}{s}}\left[\frac{1-(r_{-}/r_{+})^{n+1}}{1+(r_{-}/r_{+})^{n+1}-2r_{+}^{-n-1}}\right] \\
&\approx\sqrt{\frac{s-1}{s}}\quad,\quad s>1\quad;
\end{align*}
whereas for $s=1$, $\lim_{s\rightarrow1}\dim(\ker(H_{n}^{s}))=\frac{1}{2}(n+1)(n+2)$
and the fraction vanishes with the system's size as $2/(n+2)\approx2/n$.

We now return to the ground space of $H_{n}$.
\end{rem}
\begin{lem}
Ground space of $H=H_{n}^{s}+H_{n}^{J}$ with $J>0$ coincides with
the span of the equivalent classes $|c_{{\bf x_{n}}}\rangle$, which
are all product states. The ground space dimension is $|T_{n}|$. 
\end{lem}
\begin{proof}
The set of FF ground states of $H_{n}^{s}$ is given by Eq.~\eqref{eq:AllGroundStates}
in Lemma \ref{lem:GS_Hn}. The kernel of $J\sum_{k=1}^{n}(|0\rangle\langle0|)_{k}$
is the span of all product states of weight $n$, i.e.,
states $|{\bf t}\rangle$ where ${\bf t}\in\Sigma_{*}^{n}$. By Lemma
\ref{lem:FF_Sum_Lemma} the intersection of the two is $|c_{{\bf x_{n}}}\rangle$,
which we recall are the product states of the irreducible strings
of weight $n$ and there are $|T_{n}|$ of them (Eqs.~\eqref{eq:Tn_size}
and \eqref{eq:TnAsymp}).
\end{proof}

\subsection{Constructing good quantum codes in the ground space}

In this section we construct quantum codes that are supported on a selected subset of computational basis states that lie within the kernel of our 2-local Hamiltonian.  
We show that these quantum codes that encode a single logical qubit can have a linear distance.

Recall that the standard local spin states are $|j\rangle$, with
$j\in\Sigma=  \{-s, \dots, +s\}$. 
We also define the non-zero alphabet $\Sigma_{*}\equiv\{-s,\dots,-1,+1,\dots,s\}$.
Clearly $|\Sigma|=2s+1$ and $|\Sigma_{*}|=2s$.
Define the generalized (non-Hermitian) Pauli operators
in terms of these basis states as
\[
X=\sum_{j\in\Sigma_{c}}|j\rangle\langle j+1|\;,\qquad Z=\sum_{j\in\Sigma_{c}}\omega^{j}|j\rangle\langle j|\quad,
\]
where $\omega=\exp(2\pi i/(2s+1))$ is a root of unity,
and by $\Sigma_{c}$, we mean the set $\Sigma$ with the cyclic property that $s+1=-s$
and $-s-1=+s$. 
We denote by $X_{k}$ and $Z_{k}$ the
generalized Pauli operators that act on the $k^{\text{th}}$ spin
(qudit) and act trivially on the rest.

Recall that $\ker(H_n)$ is spanned by certain product
states of weight $n$. We denote the basis of $\ker(H_n)$ by $\mathcal{T}_n$ where 
\[
\mathcal{T}_{n}=\{|{\bf t}\rangle\;:\;{\bf t}\in T_{n}\}
\]
and $T_{n}$ is defined in Eq.~\eqref{eq:Tn}.

The quantum code that we construct will be a two-dimensional subspace of $\mathcal{T}_{n}$. In general, the logical codewords can be supported on an exponential number of basis states in $\mathcal{T}_{n}$. 
However, we select a subset of computational
basis states labels ${C}\subset T_{n}$ such that the minimum Hamming
distance of ${C}$ is at least $2t+1$, where $t$ is the designed
maximum number of correctable errors. 
The set of labels ${C}$ has the interpretation as a classical code, and we denote its distance by $\text{dist}({C})$:
\[
\text{dist}({C})\equiv\min\left\{ \text{dist}({\bf t},{\bf t}')\:|\:{\bf t},{\bf t}'\in T_{n}\:,\;{\bf t}\ne{\bf t}'\right\} \;,
\]
and
\[
\text{dist}({\bf t},{\bf t}')=\left|\left\{ t'_{i}\ne t_{i}\;:\;i\in[n]\right\} \right|
\]
is the usual Hamming distance between codewords ${\bf t},{\bf t}'\in\Sigma_{*}^{n}$.
The logical codewords of our quantum code will be supported only on the set of basis sets labeled by the classical code ${C}$.

Since the basis are product states, finding the subset ${C}$
with the desired distance can be seen as a problem in classical coding
theory.  
For example, when $s=2$ one can map the
elements of $T_{n}$ to $\mathbb{F}_{4}^{n}$ and use the properties of
quartenary codes over $\mathbb{F}_{4}^{n}$ to derive a classical
code with the prescribed distance. 
For instance, one can apply the mapping
\begin{equation}\label{eq:varphi}
\varphi(1)=0\;,\quad\varphi(-1)=1\;,\quad\varphi(2)=a\;,\quad\varphi(-2)=b\;,
\end{equation}
where  $\mathbb F_4 = \{0,1,a,b\}$,  $b=a+1$ and $a^{3}-1=0$, and 
for every ${\bf t}=t_1\dots t_n \in T_n$, define $\varphi(t_1\dots t_n) = (\varphi(t_1),\dots, \varphi(t_n) )$.

 Our strategy is to construct a code
$C$ over $\mathbb{F}_{4}^{n}$ that has a guaranteed minimum distance
and delete all codewords in it that have the forbidden substrings (01) and $(ab)$ to obtain a code $C'$. 
Then we let ${C} = \varphi^{-1}(C')$, 
 which will be the strings that define the computational (product) basis states.
 
Here, we prove that there are quantum codes within $\ker(H_n)$ that have linear distance in $n$. 
We leverage on the existence of good binary codes.
The relative distance of a code is the ratio of its distance to its length.
Good binary codes are defined as binary codes with positive relative distance.
To use the results from binary codes, we define a map $\beta$ from the binary symbols $0$ and $1$ to $2$ and $1$ respectively. 
It is then easy to see that given any binary code $C$, $\beta(C)$ is guaranteed to be a feasible subset of $T_n$, and hence we may use $\beta(C)$ to construct our quantum code.
 
Our main theorem is:
\begin{thm}\label{thm:linear-distance}
Let $0<\tau\le 1/4$ be a real and positive constant.
There exist quantum codes in $\ker(H_n)$ that encode one logical qubit and have the distance of $2 \tau n$ whenever 
\begin{align}
\ent_2(2\tau) + \ent_{2s+1}(2\tau) \log_2 (2s+1) + o(1) \le 1. \label{eq:gv-construct}
\end{align}
Second, there are explicit quantum codes which encode one logical qubit with a distance of $2 \tau n$ whenever
\begin{align}
1/2-\tau/0.11  \ge \log_2(2s+1) \ent_{2s+1}(2\tau).\label{eq:j-construct}
\end{align}
\end{thm}
\begin{rem}
We call the constructions given by optimizing \eqref{eq:gv-construct} and \eqref{eq:j-construct} as the Gilbert-Varshamov (GV)~\cite[Chpt.~1]{sloane} and Justesen construct~\cite[Chpt.~10,~Thm.~11]{sloane} respectively.
The GV construct arises from choosing a random $C$, while the Justesen construct uses the classical Justesen code to define $C$.
\end{rem}
\begin{proof}
To prove the first result, we use random coding arguments. Namely, we turn our attention to random binary codes.
By the Gilbert-Varshamov bound~\cite[Chapter 1]{sloane}, we know that such binary codes are almost surely good binary codes.
Moreover we know that for any positive integer $t$,
there exists a classical binary code $C$ that corrects $2t+1$ errors where
\begin{align}
|C| \ge  2^n / V_2(2t).
\end{align}
Using Lemma \ref{lem:main-result}, by setting ${C} = \beta(C)$, this implies that if
\begin{align}
 2^n / V_2(2t) \ge 2V_{2s+1}(2t), \label{GV-result}
\end{align}
then there exists some quantum code encoding a single qubit in $\ker(H_n)$ that also corrects $t = \tau n $ errors.
Since the inequality~\eqref{GV-result} is equivalent to the~\eqref{eq:gv-construct}, the first result of our theorem follows.

The second result follows from using Justesen's concatenated construction, which gives binary codes of with asymptotically linear distance and positive rate~\cite{sloane}[Chapter 10, Theorem 11].  
More specifically, a binary Justesen code $C_{\rm Justesen}$ is a concatenated code, with a Reed-Solomon outer code on the finite field of dimension $2^m$ as the outer code, and distinct inner codes each encoding $m$ bits into $2m$ bits.
When $C_{\rm Justesen}$ has a length of $n$, the length of the Reed-Solomon outer code is $n/(2m)$. Each inner code is a binary code and has rate $1$, and is a distinct mapping from $\mathbb F_{2^m}$ to $\mathbb F_2^{2m}$.
The relative distance $\delta = d/n$ of this family of Justesen codes is given by 
\begin{align}
    \delta \ge 0.110(1-2\log_2|C_{\rm Justesen}|/n) + o(1).
\end{align}
Rearranging this inequality and using $\tau = \delta/2$, we get 
\begin{align}
\log_2 |C_{\rm Justesen}|     \ge  n( 1/2 -  \tau/ 0.11 + o(1)) .
\label{eq:just}
\end{align}
Hence the number of codewords of a Justesen code with distance $2t+1$ is asymptotically at least $2^{n(1/2-9.1\tau + o(1))}$.
Now we set ${C} = \beta(C_{\rm Justesen})$,
and use Lemma \ref{lem:main-result} to find that a quantum code in the ground space of $\ker(H_n)$ that corrects asymptotically $\tau n$ errors exists whenever $\tau < (2s)/(2s+1)$ and the inequality \eqref{eq:j-construct} holds. 
\end{proof}
Using \eqref{eq:j-construct} and \eqref{eq:just},
we plot the attainable values of $\tau$ for different values of spins in Fig. \ref{Fig:GV-demonstration}.

In the next two sections we illustrate explicit quantum codes on 8 and 6 qudits respectively. These quantum code were obtained from punctured variants of the classical ternary Golay code, where by punctured we mean that the first three symbols of the code were ignored for the 8 qudit code, and five symbols were ignored for the 6 qudit code.

\begin{figure}
\centering
\includegraphics[scale=0.34]{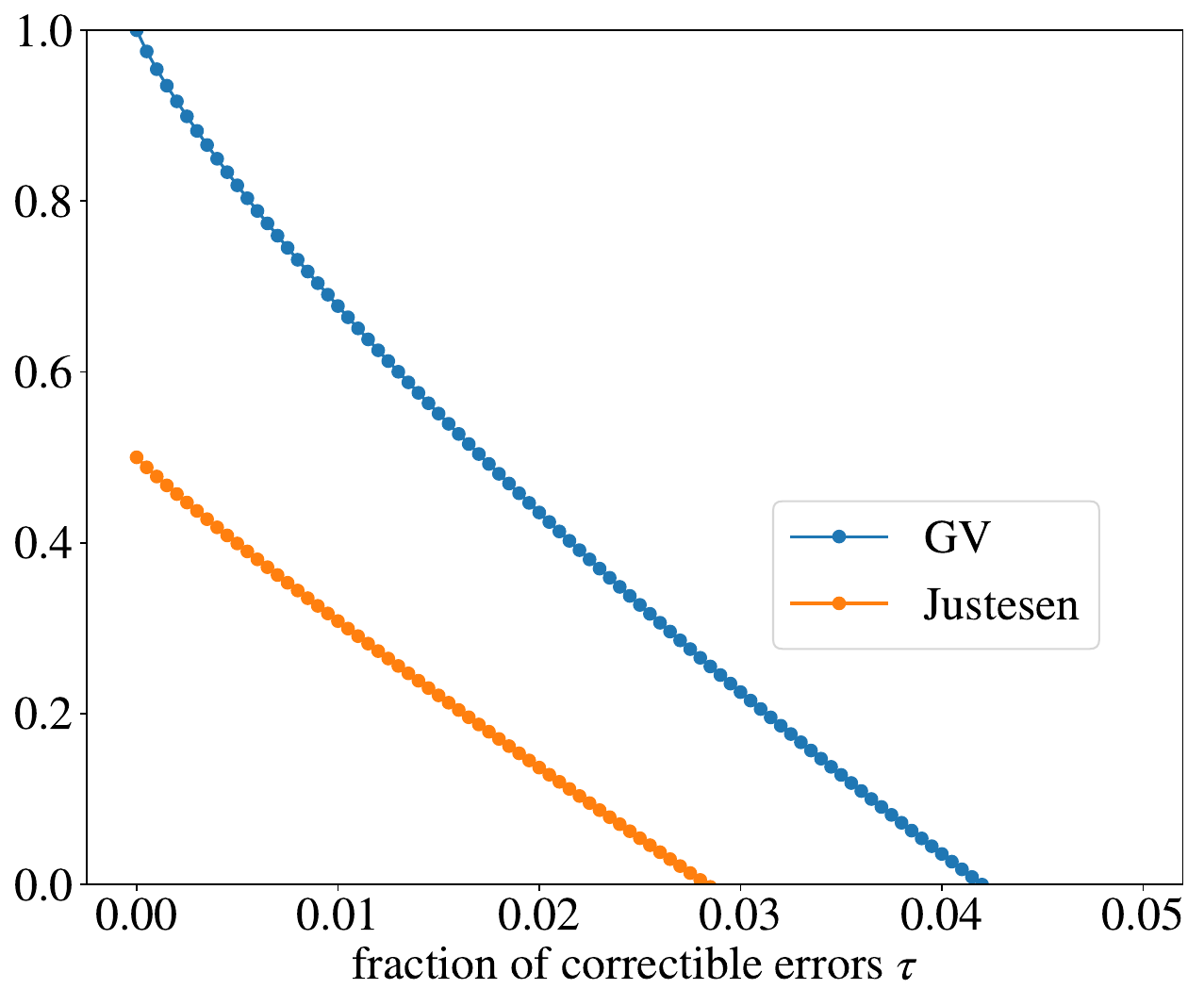} 
\caption{
The vertical axis is equal to $\dim(\ker(\log_n(|A|)))$, 
which gives a lower bound on the size of the kernel of the $A$ matrix in log-scale. The horizontal axis is $\tau = t/n$, where $t$ denotes the number of correctable errors for our quantum code in the ground space of $\ker(H_n)$. 
The length $n$ is taken to be asymptotically large. 
This demonstrates that there are linear distance quantum codes 
in the ground space of the frustration-free 2-local Hamiltonian $H_n$.}
\label{Fig:GV-demonstration}
\end{figure}

We furthermore use \eqref{eq:j-construct} and \eqref{eq:just} 
to plot the fraction of correctible errors attainable for spin chains of different spin if we use QEC codes constructed using either random classical codes or Justesen codes.

\begin{figure}
\centering
\includegraphics[scale=0.34]{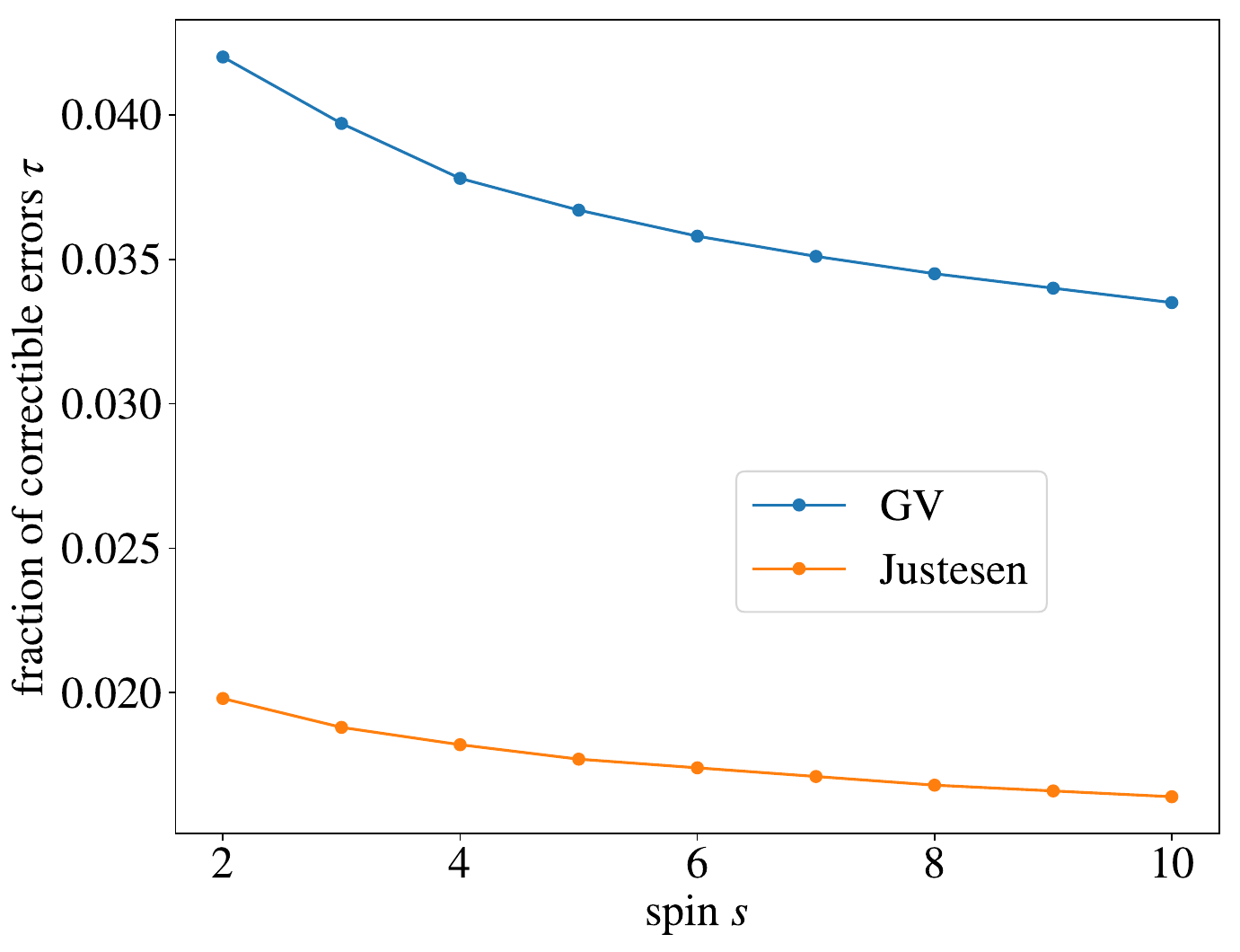} 
\caption{
The vertical axis is $\tau = t/n$, where $t$ denotes the number of correctable errors for our quantum code in the ground space of $\ker(H_n)$. The horizontal axis is the spin number our spin-chain.}
\label{Fig:spin-thresholds}
\end{figure}

\subsubsection{A ground subspace Steane code that corrects a single error
}
By slightly modifying the 7-qubit Steane code, we can embed it in the ground space of $H_7$.  Namely, on the set of computational basis vectors, we can apply the map 
$|0\> \to |1\>$ and
$|1\> \to |2\>$.
This also corresponds to the quantum code with logical codewords
\begin{align}
    |0_L\>
    =&
    \frac{1}{\sqrt{8}}\left(
    |1,1,1,1,1,1,1\>
+|1,1,1,2,2,2,2\> 
\right.\notag\\&\left.
+|1,2,2,1,2,2,1\>
+|1,2,2,2,1,1,2\>\right.\notag\\
&\left.+|2,1,2,1,2,1,2\>
+|2,1,2,2,1,2,1\>
\right.\notag\\&\left.
+|2,2,1,1,1,2,2\>
+|2,2,1,2,2,1,1\>\right)
\end{align}
and
\begin{align}
|1_L\> = &
\frac{1}{\sqrt{8}}
\left(|1,1,2,1,1,2,2\>
+|1,1,2,2,2,1,1\>
\right.\notag\\&\left.
+|1,2,1,1,2,1,2\>
+|1,2,1,2,1,2,1\>\right.\notag\\
&\left.+|2,1,1,1,2,2,1\>
+|2,1,1,2,1,1,2\>
\right.\notag\\&\left.
+|2,2,2,1,1,1,1\>
+|2,2,2,2,2,2,2\>\right).
\end{align} 
Using the KL criteria, we can verify that this quantum code corrects any single error.

\subsubsection{A ground subspace code on eight spins that corrects a single error
}
Here, we give an example of a quantum code that lies in the kernel of $H^s_n$ for $n=8$ and $s=2$. This quantum code encodes a single logical qubit, corrects an arbitrary one-qubit
error, and has logical codewords
\begin{align}
|0_L\>=& (|\phi_0\>|\theta_0\>+|\phi_1\>|\theta_1\>+|\phi_2\>|\theta_2\>
\notag\\&
+|\phi_3\>|\theta_3\>
+|\phi_4\>|\theta_4\>+|\phi_5\>|\theta_5\>)/\sqrt 6\notag\\
|1_L\>=& (|\phi_1\>|\theta_4\>+|\phi_0\>|\theta_3\>
+|\phi_3\>|\theta_0\>
\notag\\&
+|\phi_2\>|\theta_5\>+|\phi_5\>|\theta_2\>+|\phi_4\>|\theta_1\>)/\sqrt 6,
\label{eq:explicit-code}
\end{align}
where
\begin{align}
|\phi_0\> &= |1,1,1,{-}2\>,\notag\\
|\phi_1\> &= |1,{-}2,{-}1,{-}1\>,\notag\\
|\phi_2\> &= |{-}1,{-}2,{-}2,{-}1\>,\notag\\
|\phi_3\> &= |{-}1,{-}1,1,1\>,\notag\\
|\phi_4\> &= |2,{-}1,{-}1,1\>,\notag\\
|\phi_5\> &= |2,1,{-}2,{-}2\>,
\end{align}
and
\begin{align}
|\theta_0\> &= |{-}2,2,2,1\>,\notag\\
|\theta_1\> &= |1,{-}2,{-}2,{-}2\>,\notag\\
|\theta_2\> &= |{-}1,2,2,1\>,\notag\\
|\theta_3\>&= |{-}2,{-}2,{-}2,{-}2\>,\notag\\
|\theta_4\> &= |1,2,2,1\>,\notag\\
|\theta_5\> &= |{-}1,{-}2,{-}2,{-}2\>.
\end{align}
The KL criteria for correcting a single error using this quantum code are satisfied. Also this code is not a CWS code.
To see this, note that for any stabilizer state made using qudits of dimension 5 (a prime number), the number of computational basis states over which they are superposition over must be a power of 5. However the logical codewords the 8-qudit code are superpositions over 6 computational basis states, and 6 is not a power of 5. Hence there does not exist any Pauli that takes either of logical codewords to a stabilizer state.
Hence, the 8-qudit code is not a CWS code.
Therefore, we have an example of a quantum code that falls outside of the CWS, stabilizer and CSS quantum coding formalisms.

Next we point out that the quantum code in \eqref{eq:explicit-code} has a concatenated structure. The logical codewords of the outer code, given in \eqref{eq:explicit-code} are simply maximally entangled states on two six-level systems. 

From the structure of the inner codes, it is clear that to perform a logical bit-flip on our quantum code, it suffices to induce the transition $|{-}2,{-}2,{-}2\> \leftrightarrow |2,2,1\>$ on the last three spins.
Performing other logical computation operations is significantly more complicated and we leave this for future work.
 
\subsubsection{A ground subspace code that detects a single error 
}

We also construct an error detecting quantum code (distance equal to 2) on six spins with $s=2$ using the logical operators 
\begin{align}
|0_L\> &= \frac{|1,1,2,1,-2,1\> \; +\; |-2,1,-2,-2,2,2\>}{\sqrt 2}\\
|1_L\> &= \frac{|1,1,-2,-2,2,1\>\; +\;  |-2,1,2,1,-2,2\>}{\sqrt 2}.
\end{align}
Such a construction is not unique, and we have many other error detecting codes on six spins.
 
\subsubsection*{Acknowledgement}
\begin{acknowledgement}RM acknowledges funding
from the MIT-IBM Watson AI Lab under the project
{\it Machine Learning in Hilbert space}. The research was
supported by the IBM Research Frontiers Institute. YO acknowledges support from the EPSRC (Grant No. EP/M024261/1) and the QCDA project (Grant No. EP/R043825/1) which has received funding from the QuantERA ERANET Cofund
in Quantum Technologies implemented within the European Union’s Horizon 2020 Programme. YO also acknowledges support from EPSRC (Grant No. EP/W028115/1). 
\end{acknowledgement}

\bibliographystyle{quantum}
\bibliography{motzkinqecc}

\end{document}